\documentclass[footinbib]{revtex4}

\usepackage{braket}
\usepackage{amsfonts} 
\usepackage{amssymb} 
\usepackage{graphicx}  
\usepackage{amsmath}
\usepackage{amsthm}
\usepackage{enumerate}
\usepackage{color}

\usepackage{array}
\usepackage{booktabs}
\usepackage{longtable}
\usepackage{lmodern}
\usepackage[T1]{fontenc}
\usepackage[utf8]{inputenc}
\usepackage{mathtools}
\usepackage{amsmath}

\allowdisplaybreaks[1]

\def\idty{{\leavevmode\rm 1\mkern -5.4mu I}} 

\def\tr{\mathop{\rm Tr}\nolimits}

\mathchardef\ree="023C \mathchardef\imm="023D  

\newcommand{\beq}{\begin{equation}}
\newcommand{\eeq}{\end{equation}}
\newcommand{\rem}[1]{}  

\newtheorem{thm}{Theorem}[section]
\newtheorem{lem}[thm]{Lemma}
\newtheorem{prop}[thm]{Proposition}

\newtheorem{defi}[thm]{Definition}

\newtheorem{qu}[thm]{Question}

\newcommand{\sw}{\mathcal{S}} 

\renewcommand{\Bbb}[1]{\if1#1\idty\else\mathbb{#1}\fi}

\newcommand{\kb}[1]{|#1\rangle\langle#1|}

\newcommand{\SP}{\operatorname{span}}

\newcommand{\SU}{\operatorname{SU}}

\newcommand{\Sym}{\operatorname{S}}

\usepackage[utf8]{inputenc}
\usepackage[a4paper]{geometry}
\usepackage[stdtext]{youngtab}
\newcommand{\hi}{{\cal H}}
\newcommand{\xpl}{\sqrt{x}}

\newcommand{\no}{{\mathbf N}}
\newcommand{\spe}{E_{[\frac{n}{2},\frac{3n}{2}]}}
\newcommand{\speb}{E_{[\frac{n}{4},\frac{9n}{4}]}}

\usepackage{enumerate}
\usepackage[utf8]{inputenc}
\usepackage[a4paper]{geometry}
\usepackage[stdtext]{youngtab}
\newcommand{\bee}{\begin{eqnarray}}
\newcommand{\eee}{\end{eqnarray}}
\newcommand{\been}{\begin{eqnarray*}}
\newcommand{\eeen}{\end{eqnarray*}}

\newcommand{\swa}{{\mathcal S}({\cal H})}

\newcommand{\RM}{R_{M,\lambda}(x)}

\begin{document}

\title{Simulating continuous quantum systems by mean field fluctuations}

\author{Zolt\'an K\'ad\'ar}
\affiliation{Department of Pure Mathematics, University of Leeds}
\email{z.kadar@leeds.ac.uk}
\author{Michael Keyl}
\affiliation{TU M\"unchen, Fakult\"at Mathematik}
\email{michaelkeyl137@gmail.com}
\author{G\'eza T\'oth}
\affiliation{Department of Theoretical Physics, University of the Basque Country UPV/EHU}
\email{Geza.Toth.hu@gmail.com}
\author{Zolt\'an Zimbor\'as}
\affiliation{Department of Computer Science, University College London}
\email{zimboras@gmail.com}

\begin{abstract}
  In this paper we are discussing the question how a continuous quantum system
  can be simulated by mean field fluctuations of a finite number of qubits. On
  the kinematical side this leads to a convergence result  which states
  that appropriately chosen fluctuation operators converge in a certain weak
  sense (i.e. we are comparing expectation values) to canonical position and
  momentum $Q, P$ of one-degree of freedom, continuous quantum system. This
  result is substantially stronger than existing methods which rely either on
  central limit theorem arguments (and are therefore restricted to the
  Gaussian world) or are valid only if the states of the ensembles are close
  to the ``fully polarized'' state. Dynamically this relationship keeps
  perfectly intact (at least for small times) as long as the continuous system
  evolves according to a quadratic Hamiltonian. In other words we can
  approximate the corresponding (Heisenberg picture) time evolution of the
  canonical operators $Q, P$ up to arbitrary accuracy by the appropriately
  chosen time evolution of fluctuation operators of the finite systems.
\end{abstract}

\maketitle

\section{Introduction}

One of the most exciting, recent experiments in quantum information science is
the realization of quantum memory. In \cite{Polzik} it is shown that the state
of a laser field can be stored in collective degrees of freedom of an ensemble
of about $10^{23}$ earth-alcaline atoms at room temperature. When the light is
released after a few microseconds it is in a state which has a very good
fidelity with respect to the original state. 

A crucial part of the theoretical description of these experiments is the
Holstein-Primakoff transformation \cite{HP}, which basically implies that
under certain conditions the ensemble can be treated like a \emph{continuous}
quantum system  interacting harmonically with the light field. Storing and
releasing of the light is then -- roughly speaking -- just a swap of two coupled
oscillators; cf \cite{HSP} for details. 

One of the most important assumptions we have to make for this approximation
to be valid is invariance of the interaction under permutations of the atoms
(i.e. the intensity of the laser field is constant over the region where the
ensemble is spread), and this brings the whole model in contact with mean
field theory and mean field fluctuations. The Holstein-Primakoff
transformation can be regarded in this context as a special instance of the
fact that fluctuation operators (which measure the fluctuations of a mean
field variable around its expectation value) satisfy in the infinite particle
limit and under quite general conditions canonical commutation relations
\cite{Narnhofer}. 

This observation provides a big motivation to revisit the subject of mean field
fluctuations and to look how it can be used to describe the type of
matter-light interactions used in quantum memory experiments. Most of the
results available in the literature are mainly targeted towards the needs of
statistical mechanics \cite{fluct} and are strongly connected to the realm of the
central limit theorem \cite{Taku}. This implies in particular that in the
infinite particle limit we only get Gaussian states. For the type of
applications we have in mind this is a too severe restriction. To remove it
is one main goal of this paper.

To this end we are going to ask how a continuous quantum system can be
simulated by mean field fluctuations of a finite number of qubits. This
question involves kinematical as well as dynamical aspects. On the kinematical
side this means in particular to give the above statement about commutation
relations of fluctuation operators a more detailed and mathematical precise
meaning. This is done by a convergence result (Theorem \ref{thm:2}) which
states (roughly speaking) that appropriately chosen fluctuation operators
converge in a certain weak sense (i.e. we are comparing expectation values) to
canonical position and momentum $Q, P$ of one-degree of freedom, continuous
quantum system. This result is substantially stronger than existing methods
which rely either on central limit theorem arguments (and are therefore
restricted to the Gaussian world) or are valid only if the states of the
ensembles are close to the ``fully polarized'' state. As a side result we will
see that the mixedness of the one particle states (i.e. the restriction of the
global state to a single particle) is related to an effective $\hbar$ in the
corresponding continuous system.  

Dynamically this relationship keeps perfectly intact (at least for small
times) as long as the continuous system evolves according to a quadratic
Hamiltonian. In other words we can approximate the corresponding (Heisenberg
picture) time evolution of the canonical operators $Q, P$ up to arbitrary
accuracy by the appropriately chosen time evolution of fluctuation
operators of the finite systems (Theorem \ref{thm:1}). This is, again, a big
improvement over existing results where only special Hamiltonians and weaker
types of approximations are discussed.

Before we start to present the details, let us have a short look on the
organization of the paper. Section \ref{sec:basic-defin-main} collects all
main results of the paper without proofs and keeping the technical overhead as
small as possible. It consists of three subsections starting with
\ref{sec:fluct-oper} where fluctuation operators are reviewed and slightly
reformulated for the needs of this paper. Subsection \ref{sec:schwartz-operators}
discusses Schwartz operators recently introduced in \cite{SOP}. This class 
encompasses all states of the continuous system which can be simulated with
our method. In Subsection \ref{sec:perm-invar-stat} we present a new
representation of permutation invariant states of many qubit systems and
introduce a special type of convergence.  It plays a crucial role in the two
main theorem which are stated and discussed in Subsection
\ref{sec:dynam-fluct}. 
The rest of the paper is devoted to proofs
(Section \ref{sec:fluctuator-dynamics}) and an outlook to further research 
(Section \ref{sec:outlook}).

\section{Basic definitions and main results}
\label{sec:basic-defin-main}

\subsection{Fluctuation operators}
\label{sec:fluct-oper}

Let us start to setup some notations and to recall a few known results which
will be important for the rest of this paper. Therefore lets denote the space
of (bounded) operators on the Hilbert space $\mathcal{H}=\Bbb{C}^2$ by
$\mathcal{B}(\mathcal{H})$. For each $A \in \mathcal{B}(\mathcal{H})$ and a
fixed density operator $\theta$ (the \emph{reference state}) on $\mathcal{H}$
define the \emph{fluctuation operator} $F_M(A) \in
\mathcal{B}(\mathcal{H}^{\otimes M})$ by 
\begin{equation}
  F_M(A) = \frac{1}{\sqrt{M}} \sum_{i=1}^M \left( A^{(i)} - \tr(A\theta) \Bbb{1}\right)
\end{equation}
where
\begin{equation}
  A^{(i)} = \Bbb{1}^{\otimes (i-1)} \otimes A \otimes \Bbb{1}^{M-i}.
\end{equation}
The $F_M(A)$ are important in mean field theory and describe the quantum
fluctuations around a mean field observables; cf \cite{fluct} for
details. For us the $F_M(A)$ will serve as approximations (in a sense we will
make clear below) of canonical operators of a continuous quantum system. To
this end the central definition is the following:

\begin{defi} \label{def:1}
  A sequence of density operators $\rho_M$, $M \in \mathbb{N}$ has
  $\sqrt{M}$\emph{-fluctuations} if 
  \begin{equation} \label{eq:5}
    \lim_{M \rightarrow \infty} \tr\bigl(\rho_M F_M(A_1) \cdots F_M(A_K)\bigr)
  \end{equation}
  exists and is finite for all $K \in \mathbb{N}$ and all $A_1, \dots, A_K \in 
  \mathcal{B}(\mathcal{H})$.
\end{defi}

For the rest of this paper let us choose a reference state $\theta$ satisfying
(where $\sigma_\alpha$, $\alpha=1,2,3$ denote Pauli matrices)
\begin{equation} \label{eq:2}
  \tr(\theta \sigma_1) = \tr(\theta \sigma_2) = 0,\quad \tr(\theta \sigma_3) =
  2\lambda,\ \lambda \in [0,1]
\end{equation}
and consider fluctuation operators corresponding to $\sigma_1$ and $\sigma_2$.
\begin{equation} \label{eq:7}
  Q_M = F_M\left(\frac{\sigma_1}{\sqrt{2}}\right) = 
  \frac{L_{M,1}}{\sqrt{2 M}} \quad
  P_{M} = F_M\left(\frac{\sigma_2}{\sqrt{2}}\right) = 
  \frac{L_{M,2}}{\sqrt{2 M}}
\end{equation}
here $L_{M,\alpha}$ denote global angular momentum (or global spin) operators
given by
\begin{equation}
  L_{\alpha,M} = \frac{1}{2} \sum_i \sigma_{\alpha}^{(i)}, \quad \alpha=1,2,3.
\end{equation}
The operators $Q_M$ and $P_M$ satisfy the commutation relations
\begin{equation} \label{eq:10}
  [Q_M, P_M] - i\lambda \Bbb{1} = \frac{i}{2 \sqrt{M}} F_N(\sigma_3). 
\end{equation}
Taking expectation values and the limit $M \rightarrow \infty$ on both sides
leads to
\begin{equation}
  \lim_{M\rightarrow\infty} \tr\bigl(\rho_M [Q_M,P_M]\bigr) -i \lambda =
  \frac{i}{2\sqrt{M}} \lim_{M \rightarrow \infty} \tr(\rho_M
  F_M\bigl(\sigma_3)\bigr) = 0.
\end{equation}
This indicates that Equation (\ref{eq:10}) converges in a certain formal way
towards the canonical commutation relations. A more precise version of this
statement is the following (the proof is postponed to the appendix):

\begin{prop} \label{prop:1}
  Consider a sequence $\rho_M \in \mathcal{B}(\mathcal{H}^{\otimes M})$, $M \in
  \Bbb{N}$ of density operators with $\sqrt{M}$ fluctuations. Then there is a
  Hilbert space $\mathcal{H}_\infty$, a density operator $\rho_\infty$ and two
  symmetric operators $Q_\infty, P_\infty$ with the common, invariant, dense
  domain $D$ such that for any polynomial $f(q,p)$ in two non-commuting
  variables $q,p$ the following statements hold:
  \begin{enumerate}
  \item \label{item:1}
    $\rho_\infty D \subset D$,
  \item \label{item:2}
    The closures of $f(Q_\infty,P_\infty) \rho_\infty$ and $\rho_\infty
    f(Q_\infty,P_\infty)$ are of trace class. 
  \item \label{item:3}
    $\lim_{M \rightarrow \infty} \tr\bigl(f(Q_M,P_M) \rho_M\bigr) = \tr\bigl(f(Q_\infty,P_\infty)
    \rho_\infty\bigr) = \tr\bigl(\rho_\infty f(Q_\infty,P_\infty)\bigr)$,
  \item \label{item:4}
    $[Q_\infty,P_\infty] \phi = i \lambda \phi$ $\forall \phi \in D$.
  \end{enumerate}
\end{prop}

Condition \ref{item:1} and \ref{item:2} are of technical nature and needed to
guarantee that \ref{item:3} is a mathematically well defined expression. We
come back to this point in the next subsection. Item \ref{item:3} itself
establishes the operators $Q_\infty$ and $P_\infty$ as a form of weak limit of the
fluctuators $Q_M, P_M$, and condition \ref{item:4} states that these limit
operators satisfy the canonical commutation relations with $\lambda$
from Equation (\ref{eq:2}) \emph{as an effective $\hbar$}. We are tempted to
interpret this result by saying that ``the collective fluctuations of an
ensemble of two-level atoms behave in the infinite particle limit like a
continuous quantum system with one degree of freedom.'' Unfortunately this
interpretation is premature, since condition \ref{item:4} is not sufficient to
conclude that $Q_\infty$ and $P_\infty$ are canonical position and momentum. It is not even
guaranteed that $Q_\infty$ and $P_\infty$ are essentially self-adjoint on the domain
$D$. Hence the crucial question is:

\begin{qu} \label{qu:1}
  For which sequences $\rho_M \in \mathcal{B}(\mathcal{H}^{\otimes M})$, $M \in
  \Bbb{N}$ Proposition \ref{prop:1} holds with
  \begin{enumerate}
  \item $\mathcal{H}_\infty = \mathrm{L}^2(\Bbb{R})$ -- i.e. the Hilbert space
    of square integrable functions on $\Bbb{R}$,
  \item $D = \mathcal{S}(\Bbb{R})$ -- i.e. the space of Schwartz functions on
    $\Bbb{R}$, 
  \item \label{item:8}
    $Q_\infty= Q$ with $(Q \phi)(x) = x \phi(x)$ $\forall \phi \in
    \mathcal{S}(\Bbb{R})$ -- i.e. $Q$ is the canonical position operator with
    $\mathcal{S}(\Bbb{R})$ as its domain. 
  \item \label{item:9}
    $P_\infty=\lambda P$ with $P \phi = -i \phi'$ $\forall \phi \in
    \mathcal{S}(\Bbb{R})$ -- i.e. $P$ is the canonical momentum operator with
    $\mathcal{S}(\Bbb{R})$ as its domain. 
  \end{enumerate}
\end{qu}

One way to answer this question is to be more restrictive with the selection
of the sequences of states $\rho_M$, $M \in \Bbb{N}$. An important case arises
if  take a translationally invariant state $\omega$ of an infinite spin-chain
and define the $\rho_M$ in terms of the restrictions of $\omega$ to a sub
chain of length $M$. It is shown \cite{Taku} that in the case where
$\omega$ is exponentially clustering (i.e. correlations decay exponentially
fast), Proposition \ref{prop:1} holds with $Q_\infty$ and $P_\infty$ being the canonical
position and momentum operators and $\rho_\infty$ being a Gaussian state. 

In this paper, we want to discuss a slightly different approach which puts more
emphasis on the non-Gaussian world. In the following two sections we will
present in particular a different form of convergence which is tighter than
pointwise convergence (as used in Definition \ref{def:1}), and which can
therefore guarantee that the $M \rightarrow \infty$ limit always leads to the
Schrödinger representation of the CCR.  

Before we come to this point let us add a short remark about the role of the
reference state $\theta$. To this end let us calculate first the expectation
value  of one fluctuation operator $F_M(A)$. For an arbitrary $\rho_M$ we have
\begin{equation} \label{eq:1}
  \tr(A^{(i)} \rho_M) = \tr(A \theta_{M,i}) 
\end{equation}
where $\theta_{M,i}$ is the restriction of $\rho_M$ to the $i^{\mathrm{th}}$ site
(for a permutation invariant $\rho_M$ all $\theta_{M,i}$ coincide). Hence
\begin{equation}
  \tr(F_M(A) \rho_M) = \sqrt{M} \tr\bigl((\theta_M - \theta) A\bigr),\quad \theta_M =
  \frac{1}{M} \sum_{i=1}^M \theta_{M,i}. 
\end{equation}
Hence, $\lim_{M\rightarrow \infty} \tr(F_M(A)\rho_M) < \infty$ for all $A \in
\mathcal{B}(\mathcal{H})$ can hold only if
\begin{equation}
  \lim_{M\rightarrow\infty} \theta_M = \theta
\end{equation}
is satisfied. In other words: if a sequence $\rho_M$ has $\sqrt{M}$
fluctuations, its averaged one-site restriction in the limit $M \rightarrow
\infty$ is given by $\theta$. With the choice from Equation (\ref{eq:2}) this
implies that the parameter $\lambda$  -- which plays the role of an
\emph{effective $\hbar$} in Condition \ref{item:4} of Proposition \ref{prop:1}
-- is a property of \emph{the sequence}, while the operators $Q_M, P_M$ do not
depend on $\lambda$.

\subsection{Schwartz operators}
\label{sec:schwartz-operators}

Conditions \ref{item:1} and \ref{item:2} of Proposition \ref{prop:1} are
needed to guarantee that the traces on the right hand site of the expression
in item \ref{item:3} are well defined. They are, however, quite restrictive
and lead to a new class of operators which are called Schwartz operators.

\begin{defi}
  An operator $\rho \in \mathcal{B}(\mathcal{H}_\infty)$ is called a
  \emph{Schwartz operator} if $P^\alpha Q^\beta \rho Q^{\beta'} P^{\alpha'}$ is
  for all $\alpha,\alpha',\beta,\beta' \in \Bbb{N}_0$ a trace-class
  operator. The set of all Schwartz operators will be denoted by
  $\mathcal{S}(\mathcal{H}_\infty)$. 
\end{defi}

Schwartz operators are introduced and discussed in \cite{SOP}. For our purpose
only few properties are needed which we will present in the following. For
more details (in particular and the many degrees of freedom version) we will
refer the reader to \cite{SOP}. Most important for us is the fact that
$\mathcal{S}(\mathcal{H}_\infty)$ admits a natural topology.

\begin{prop}
  For each $\alpha,\alpha',\beta,\beta' \in \Bbb{N}_0$ the functional
  \begin{equation} \label{eq:16}
    \mathcal{S}(\mathcal{H}_\infty) \ni \rho \mapsto \parallel
    \rho \parallel_{\alpha\alpha'\beta\beta',1} = \parallel P^\alpha Q^\beta
    \rho Q^{\beta'} P^{\alpha'} \parallel_1 \in \Bbb{R}^+
  \end{equation}
  is a seminorm. $\mathcal{S}(\mathcal{H}_\infty)$ together with this family
  is a Fr\'echet space.
\end{prop}

In passing let us add the remark that we could replace the trace norm in
(\ref{eq:16}) by any other $p$-norm (with $p>1$ including the operator norm)
and we would still get valid families of semi-norms (including in particular
the property that all the seminorms are finite for all Schwartz operators)
defining the same topology as the family we have chosen; cf \cite{SOP} for
details. 

There are a number of alternative characterizations of Schwartz operators and
some of them we will encounter later in this paper. For now we only need a
statement about some matrix elements. To this end let us introduce some
notation first. The usual creation and annihilation operators are
\begin{equation}
  a = \frac{1}{\sqrt{2}}(Q+iP),\quad a^\dagger = \frac{1}{\sqrt{2}} (Q-iP)
\end{equation}
both with $\mathcal{S}(\Bbb{R})$ as their domain. They give rise to the number
operator\footnote{The overline means -- as usual -- the closure.} 
\begin{equation} \label{eq:17}
  \mathbf{N} = \overline{a^\dagger a},\quad \mathbf{N} \psi_n = n \psi_n,\quad
  n \in \Bbb{N}_0
\end{equation}
and its eigenfunction $\psi_n$, $n \in \Bbb{N}_0$, i.e. $\psi_n$ denotes the
$n^{\mathrm{th}}$ order \emph{Hermite function}. Now we have the following
proposition \cite{SOP}.

\begin{prop} \label{prop:3}
  A trace class operator $\rho$ is a Schwartz operator iff
  \begin{equation}
    \sup_{n,m \in \Bbb{N}_0} \lvert \langle \psi_n, \rho \psi_m\rangle \rvert
    \cdot (|n| + |m|)^k < \infty \quad \forall k \in \Bbb{N}_0.
  \end{equation}
\end{prop}

This implies in particular that all operators $\rho$ where only a finite
number of matrix elements $\langle \psi_n, \rho \psi_m\rangle$ are different
from $0$, are Schwartz operators.

\subsection{Permutation invariant states}
\label{sec:perm-invar-stat}

Let us come back now to Question \ref{qu:1}. The purpose of this section is to
define a notion of convergence which guarantees that the operators $Q_\infty$ and
$P_\infty$ in Proposition \ref{prop:1} can be chosen as canonical position
and momentum $Q$ and $\lambda P$. To this end let us have a another look at Definition
\ref{def:1}. The crucial objects in Equation (\ref{eq:5}) are expectation
values of products of fluctuation operators $F_M(A_M)$. The $F_M(A)$, however,
commute with all permutation operators. Therefore the expectation values of
the form  $\tr(F_M(A_1) \cdots F_M(A_K) \rho_M)$ do not change, if we replace
the density operator $\rho_M \in \mathcal{B}_*(\mathcal{H}^{\otimes M})$ with
its average over all permutations. In other words: As long as we are only
interested in expectation values we can assume without loss of generality that
$\rho_M$ is permutation invariant.

The general structure of permutation invariant density matrices can be easily
deduced from representation theory of $\SU(2)$ and the permutation group (this
is very well known stuff; cf. \cite{cleanex} for a review and further
references). The $M$-fold Hilbert space $\mathcal{H}^{\otimes M}$ can be
decomposed such that for any $U \in \SU(2)$ we have
\begin{equation}
  \mathcal{H}^{\otimes M} = \bigoplus_j \mathcal{H}_j \otimes
  \mathcal{K}_{M,j},\quad U^{\otimes M} = \pi_j(U) \otimes \Bbb{1}
\end{equation}
where $\mathcal{H}_j$ is a $2j+1$ dimensional Hilbert space, $\pi_j: \SU(2)
\rightarrow \mathcal{B}(\mathcal{H}_j)$ denotes the irreducible spin-$j$
representation, $\mathcal{K}_{M,j}$ is a multiplicity space carrying an
irreducible representation of the permutation group $\Sym_M$, and the index
$j$ runs over the integers $0, 1, \dots M/2$ (if $M$ is even) or the
half-integers $1/2, 3/2, \dots, M/2$ (if $M$ is odd). A permutation invariant
density matrix has the form
\begin{equation}\label{eq:3}
  \rho_M = \bigoplus_j w_{M,j} \left(\rho_{M,j} \otimes \frac{\Bbb{1}}{\dim
      \mathcal{K}_{M,j}}\right) 
\end{equation}
where $\rho_{M,j}$ is a density matrix on $\mathcal{H}_j$ and the weights
$w_{M,j}$ are positive real numbers satisfying
\begin{equation}
  \sum_j w_{M,j} = 1.
\end{equation}
If $w_{M,j} = 0$ for some $j$ we set $\rho_{M,j} = 0$, too. 

In a similar way we can decompose the angular momentum operators $L_{\alpha,N}$.
 If we write\footnote{In abuse of notation we use the
  symbol for the representation of the group and the corresponding
  Lie-algebra.} $L_\alpha^{(j)} = \pi_j(\sigma_\alpha/2) \in
\mathcal{B}(\mathcal{H}_j)$ for angular momentum in the spin-$j$
representation we get 
\begin{equation}
  L_{\alpha,N} = \bigoplus_j L_\alpha^{(j)} \otimes \Bbb{1}.
\end{equation}
Of special importance as well are the eigenvectors $\psi^{(j)}_n$ of
$L_3^{(j)}$
\begin{equation}
  L_3^{(j)} \psi^{(j)}_n = (j - n) \psi^{(j)}_n, \quad n=0, \dots, 2j
\end{equation}
and the ladder operators $L_{\pm}^{(j)} = L_1^{(j)} \pm i L_2^{(j)}$
\begin{equation}
  L^{(j)}_+ \psi^{(j)}_n = \sqrt{n(2j-n+1)} \psi^{(j)}_{n-1} \quad
  L^{(j)}_- \psi^{(j)}_n = \sqrt{(2j-n)(n+1)} \psi^{(j)}_{n+1}.
\end{equation}

Th next important points concerns the relation between the Hilbert spaces 
$\mathcal{H}_j$ and $\mathcal{H}_\infty$. The former is the representation
space of the spin-$j$ irreducible $\SU(2)$ representation and contains the
distinguished basis $\psi_n^{(j)}$. The latter was introduced in Subsection
\ref{sec:schwartz-operators} as  $\mathrm{L}^2(\Bbb{R})$ and we chose the
\emph{Hermite functions} $\psi_n$, $n \in \Bbb{N}$ as a basis (cf. Equation
(\ref{eq:17})). Now we embed all the $\mathcal{H}_j$ into $\mathcal{H}_\infty$
such that $\psi_n^{(j)}$ becomes for all $n=0,\dots,2j$ the Hermite function
$\psi_n$. In other words, we consider the embedding:
\begin{equation}
  \mathcal{H}_j \ni \sum_{n=0}^{2j} \phi^n \psi_n^{(j)} \mapsto
  \sum_{n=0}^{2j} \phi^n \psi_n \in \mathcal{H}_\infty.
\end{equation}
In the following we will \emph{identify $\mathcal{H}_j$ with its image under
  this isometry}. In this way all operators $\rho_M$ become finite rank
operators on the same infinite dimensional Hilbert space
$\mathcal{H}_\infty$. The same is true for the angular momentum operators
$L_\alpha^{(j)}$, $\alpha=1,2,3$ and $L_\pm^{(j)}$.

Another crucial step is a rescaling of the parameter $j = 0, \dots, M/2$ (or
$1/2,\dots,M/2$) by 
\begin{equation} \label{eq:11}
  x_j = \frac{2j}{M}\quad j_x = \frac{\lfloor M  x \rfloor}{2}.
\end{equation}
Hence $\rho_{M,j}$ becomes $\rho_{M,x_j}$ and if we extend it continuously to
the interval $[0,1]$ we get a continuous function
\begin{equation}
  R_M : [0,1] \ni x \mapsto R_M(x) \in \sw(\mathcal{H}_\infty),\quad
  R_M(x_j) = \rho_{M,j}.
\end{equation}
Note that $R_M$ is not uniquely defined by $\rho_M$, however we can choose
such a function for any $\rho_M$ (e.g. by linear interpolation). Also note
that the definition of these functions requires the embedding of the
$\mathcal{H}_j$ into $\mathcal{H}_\infty$ as described above.

The advantage of this formulation just introduced is that sums over $j$ can
now be reformulated as integrals with respect to the probability measure
$\mu_M$ given by
\begin{equation} \label{eq:13}
  \int_0^1 f(x) \mu_M(dx) = \sum_j w_{M,j} f(x_j). 
\end{equation}
To see this consider the expectation value of a permutation invariant
observable (like angular momentum introduced above) $A_M = \sum A_{M,j}
\otimes \Bbb{1}$ in the state $\rho_M$. Using the decomposition in Equation
(\ref{eq:3}) we get \begin{equation} \label{eq:4}
  \tr(\rho_M A_M) = \sum_j w_{M,j} \tr(\rho_{M,j} A_{M,j}) = \int_0^1
  \tr(R_M(x) A_M(x)) \mu_M(dx)
\end{equation}
where we have used again a continuous function $A_M: [0,1] \rightarrow 
\mathcal{B}(\mathcal{H}_\infty)$ satisfying $A_M(x_j) = A_{M,j}$. This
motivates the following definition:

\begin{defi}
  Given $\rho_M$, permutation invariant. We call a continuous function
  $R_M:[0,1]\to \sw(\mathcal{H}_\infty)$ an \textbf{integral representation}
  of $\rho_M$, if $R_M(2j/M) = \rho_{M,j}$, where $\rho_{M,j} \in
  \mathcal{B}_*(\mathcal{H}_j)$ are the density operators in the decomposition
  from Equation (\ref{eq:3}).
\end{defi}
 
Note again that an integral representation always exist, but is never
unique. The latter is not a problem, though, since integrals as in Equation
(\ref{eq:4}) only depend on $\rho_M$ and not on the integral representation
chosen. 

Finally let us have a look at the convergence of sequences of permutation
invariant density operators $\rho_M$, $M \in \Bbb{N}$. We can split this
question up into two pieces: convergence of the sequence of measures
$\mu_{\rho,M}$ and convergence of the sequence of functions $R_M$. A
particular behavior is described by the next definition, which captures the
situation that in the limit $M \rightarrow \infty$ only one value $x = \lambda
\in [0,1]$ is relevant, while all other $x$ does not matter. Or to state it in
a different way: For large $M$ the state $\rho_M$ describes a very sharp
angular momentum $\vec{L}^2$ with $j$ centered around $j_\lambda = \lambda
M/2$. For these $j$ the corresponding $\rho_{M,j}$ is close to $\rho_\infty$
(in the topology of $\sw(\mathcal{H}_\infty)$ which measures in turn the
moments). All other $\rho_j$ are unimportant.

\begin{defi} \label{def:2}
  We say that a sequence $\rho_M \in \mathcal{B}_*(\mathcal{H}^{\otimes M})$
  of permutation invariant density operators \emph{converges at $\lambda \in
    [0,1]$ towards} a state $\rho_\infty \in \sw(\mathcal{H}_\infty)$ if each
  $\rho_M$ admits an integral representation $R_M$ such that the following
  conditions hold:  
  \begin{enumerate}
  \item 
    The sequence of probability measures $\mu_M$ converges weakly to the
    point measure to the point measure at $\lambda$ that is 
    \begin{equation}
      \lim_{M \rightarrow \infty} \int_0^1 f(x) \mu_M(dx) = f(\lambda)
    \end{equation}
    for all continuous functions $f:[0,1]\to \mathbb{R}$ 
  \item 
    The set 
    \begin{equation}
      \{R_M(x)|M\in\mathbb{N}, x\in [0,1]\} \subset \sw(\mathcal{H}_\infty)
    \end{equation}
    is bounded.
  \item There is a neighbourhood of $I\subset [0,1]$ of $\lambda$ and a
    continuous function $R_\infty: I \rightarrow \sw(\mathcal{H}_\infty)$
    such that $R_\infty(\lambda) = \rho_\infty$ and 
    \begin{equation}
      \lim_{M\to \infty}R_M(x)=R_\infty(x) \;\mbox{uniformly on}\;I 
    \end{equation}
    holds in the topology of $\sw(\mathcal{H}_\infty)$.
  \end{enumerate}
\end{defi}

\subsection{The main result}
\label{sec:dynam-fluct}

Now we can come back to Question \ref{qu:1}. The next theorem simply says that
the convergence given in Definition \ref{def:2} provides (at least) a
sufficient condition (all proofs are postponed to Section
\ref{sec:fluctuator-dynamics}). 

\begin{thm} \label{thm:2}
  For any sequence $\rho_M \in \mathcal{B}_*(\mathcal{H}^{\otimes M})$, $M
  \in \Bbb{N}$ of permutation invariant states converging towards a
  $\rho_\infty \in \sw(\mathcal{H}_\infty)$ at $\lambda \in (0,1]$, and all
  polynomials $f(q,p)$ in two non-commuting variables we have 
  \begin{equation}
    \lim_{M \rightarrow \infty} \tr(\rho_M f(Q_M,P_M)) = \tr(\rho_\infty
    f(Q_\infty,P_\infty)).
  \end{equation}
  with $Q_\infty = \sqrt{\lambda}Q$ and $P_\infty=\sqrt{\lambda} P$.
\end{thm}

Note that convergence at $\lambda$ in the sense of Definition \ref{def:2} \emph{does
not imply} $\sqrt{M}$ fluctuations. The reason is that we are not controlling
convergence of the moments of the measures $\mu_M$. Therefore expectation
values of products involving $F_M(\sigma_3)$ can still diverge. It is easy to
fix this by considering a stronger type of convergence for the
$\mu_M$. However, this is not needed for getting the proper limit of the $Q_M,
P_M$. 

One way to interpret Theorem \ref{thm:2} is in terms of \emph{quantum
  simulation}: If we are interested in the expected value of an observable
$X=f(Q_\infty,P_\infty)$ of a continuous quantum system in a state $\rho_\infty$, we can
measure instead the observable $X_M=f(Q_M,P_M)$ of an ensemble of $M$ qubits
in the state $\rho_M$ and the error we make that way can be arbitrarily small,
provided $M$ is big enough. 

It is possible to simulate each density operator $\rho_\infty \in
\mathcal{S}(\mathcal{H}_\infty)$ that way. To see this consider for each half-integer
$j$ the projection
\begin{equation} \label{eq:21}
  E_j : \mathcal{H}_\infty \rightarrow \SP \{\psi_n\,|\, 0 \leq n \leq 2j\} 
\end{equation}
and a sequence $j_M \in \Bbb{R}$ satisfying $j_M \in \{0, \dots, M/2\}$ if $M$
is even, $j_M \in \{1/2, \dots, M/2\}$ if $M$ is odd, and $\lim_{M \rightarrow
  \infty} 2j_M/M = \lambda$. Now we can define
\begin{equation}
 \tilde{\rho}_M = E_{j_M} \rho_\infty E_{j_M} \otimes \Bbb{1} \in
  \mathcal{B}(\mathcal{H}_{j_M} \otimes \mathcal{K}_{M,j_M}) \subset
  \mathcal{B}(\mathcal{H}^{\otimes M})
\end{equation}
and 
\begin{equation}
  \rho_M = \frac{\tilde{\rho}_M}{\tr(\tilde{\rho}_M)} \quad \text{if
    $\tr(\tilde{\rho}_M) \neq 0$,}
\end{equation}
where we have used the identification of $\SU(2)$ representation space
$\mathcal{H}_{j_M}$ with the subspace $E_{j_M}\mathcal{H}_\infty$ introduced
in Subsection \ref{sec:perm-invar-stat}. If $\tr(\tilde{\rho}_M) = 0$ holds
(this can only happen for finitely many $M$) we can just choose an arbitrary
density operator. It is obvious that the sequence just constructed converges
towards $\rho_\infty$ at $\lambda$ and therefore it simulates $\rho_\infty$ in
the sense described in the last paragraph. 

The parameter $\lambda$ can be freely chosen in this construction (by choosing
the sequence $j_M$) without affecting the simulability of $\rho_\infty$. This
seems to indicate that $\lambda$ does not play an important role, but this
impression is wrong. At the end of Subsection \ref{sec:fluct-oper} we have
seen that $\lambda$ describes (asymptotically, i.e. in the limit $M
\rightarrow \infty$) the noise in the one-site restrictions of the $\rho_M$,
ranging from ``fully polarized'' for $\lambda=1$ and ``fully depolarized'' for
$\lambda=0$. 

At the same time it plays in the commutation relations in Proposition
\ref{prop:1} the role of an effective $\hbar$, and this also holds for the
operators $Q_\infty$ and $P_\infty$ in Theorem \ref{thm:2}. However, in contrast
to our expectation in Question \ref{qu:1} $Q_\infty$ and $P_\infty$ are not
exactly canonical position and momentum (except in the case $\lambda=1$), they
differ by a factor $\lambda^{-1/2}$ and $\lambda^{1/2}$ respectively. Hence,
although $\lambda$ has \emph{some} properties of $\hbar$ the limit $\lambda
\rightarrow \infty$ is not the classical limit (unless we rescale the finite
dimensional observables $Q_M$ and $P_M$ appropriately). Instead all observables
vanish if $\lambda \rightarrow \infty$. Nevertheless, the case $\lambda=0$ is
interesting (in particular because it is important experimentally \cite{T})
and therefore deserves a more detailed study in a forthcoming paper. We come
back to this point in Section \ref{sec:outlook}. 

Now let us come back to the simulation point of view. We have argued that we
can simulate states, but what about time-evolutions? To answer this question,
consider a self-adjoint\footnote{For each polynomial $f(q,p)$ in two
  non-commutative variable we can construct the adjoint we apply complex
  conjugation to the coefficients and reverse the ordering of each monomial. A
  polynomial is self-adjoint, if it coincides with its adjoint.} polynomial
$h(q,p)$ and define the Hamiltonians $H_M = h(Q_M,P_M)$ and
$H=h(Q,P)$. Theorem  \ref{thm:2} shows that we can simulate expectation values
of $H$ by expectation values of $H_M$. But what about time-evolutions? This
question is much more difficult. The first big obstacle arises from the fact
that $H$ (with $\mathcal{S}(\Bbb{R})$ as its domain) is only symmetric but in
general not essentially selfadjoint. It is even possible that no self-adjoint
extensions exist at all \cite{RESI2}. In these cases the unitaries 
\begin{equation} \label{eq:6}
  U_{M,t} = \exp(-i t H_M)
\end{equation}
(which always exist, since the underlying Hilbert space is finite dimensional)
do approximate a reasonable time-evolution of the continuous system. However,
what about the case where $H$ is essentially self-adjoint on
$\mathcal{S}(\Bbb{R})$? In this case we can define the weakly continuous
one-parameter group 
\begin{equation} \label{eq:8}
  U_{\lambda,t} = \exp(-i t \lambda H)
\end{equation}
and ask whether the finite dimensional dynamics $U_{M,t}$ approximates the
$U_{\lambda,t}$ with a certain error for a finite amount of time. The next
theorem shows that this indeed the case for quadratic Hamiltonians. 

\begin{thm} \label{thm:1}
  Consider: A sequence $\rho_M \in \mathcal{B}_*(\mathcal{H}^{\otimes M})$, $M
  \in \Bbb{N}$ of permutation invariant states converging towards a
  $\rho_\infty \in \sw(\mathcal{H}_\infty)$ at $\lambda \in (0,1]$, a second
  order, self adjoint, homogeneous polynomial $h(q,p)$ and the
  corresponding unitaries, as defined in Equations (\ref{eq:6}) and
  (\ref{eq:8}). Then there is $t_0 > 0$ such that for all $t\in \Bbb{R}$ with
  $|t| < t_0$ and all polynomials $f(q,p)$ we have
  \begin{equation}
    \lim_{M \rightarrow \infty} \tr(U_{M,t} \rho_M U_{M,t}^* f(Q_{M},P_{M}))
    = \tr(U_{\lambda,t} \rho_\infty U_{\lambda,t}^* f(Q_\infty,P_\infty)).
  \end{equation}
  with $Q_\infty = \sqrt{\lambda} Q$ and $P_\infty = \sqrt{\lambda} P$. 
\end{thm}

This theorem implies that the time evolution of the expectation values
$\langle f(Q,P)\rangle$ of any observable $f(Q,P)$ can be approximated for
short times and with a bounded error $\epsilon$ by the time evolution of the
corresponding expectation value $\langle f(Q_M,P_M)\rangle$ of the finite
dimensional system. The size of the error can be as small as possible provided
the number $M$ of qubits is big enough. This strengthens the simulation point
of view introduced above (at least as long as quadratic Hamiltonians are
considered). 

There is, however, one small problem with this result, and this is the time
limit $t_0$. As long as we are considering an ensemble of \emph{fixed} size
$M$ it is natural that we can keep a bounded error $\epsilon$ only for a
finite amount of time. The bound $t_0$ in theorem, however, holds also
\emph{in the infinite particle limit}. This is counter intuitive since we
would assume that we can always improve the quality the simulation by
increasing the number $M$ of particles. Therefore, it is likely that the
restriction of a finite $t_0$ is only a restriction of the proof and not of
the result.


\section{Fluctuator dynamics}
\label{sec:fluctuator-dynamics}

The purpose of this section is to prove Theorem \ref{thm:1} and in this
context we will learn more details about the dynamics of fluctuation
operators. As another side effect, we will get a proof of Theorem \ref{thm:2}
as well, since it arises from \ref{thm:1} in the special case $t=0$. Let us
start with some additional notations.
\begin{itemize}
\item 
  We need the functions
  \begin{multline} \label{eq:23}
    \beta: \mathbb{N}^2\times [0,1]\to [0,1]: (M,n,x)\mapsto \beta_{M}(x,n) =
    \theta_1(x -\frac{n}{M\lambda})\\ \text{with}\
    \theta_1(x)=\sqrt{x\chi_{[0,1]}(x)},     
  \end{multline}
  where $\chi_S$ is the characteristic function of the set $S$. 
  Some of the arguments of $\beta$ will be often omitted
  when confusion can be avoided.
\item 
  Furthermore we introduce $a_M: [0,1]\to {\cal S}({\cal H}_\infty), \forall M\in{\mathbb N}$ by
  \begin{eqnarray}
  a_M(x)&=&\beta_{M}(x,\mathbf{N})a=a\beta_{M}(x,\mathbf{N}-\Bbb{1})\label{b1}\\
  a^*_{M}(x)&=&\beta_{M,\lambda}(x,\mathbf{N}-\Bbb{1})a^*=a^*\beta_{M}(x,\mathbf{N})\label{b2}
  \end{eqnarray}
  where $a,a^*$ are the standard annihilation and creation operators,
  respectively, and $\mathbf{N}=a^* a$ is the number operator (cf. Section
  \ref{sec:schwartz-operators}) and $\beta$ with an operator argument is
  understood in the sense of functional calculus. 
\item 
  We will also need the integrated version of $a_M(x)$, $a_M^*(x)$, which are
  given by 
  \begin{equation}
    a_{M} = \frac{L_{M,+}}{\sqrt{M}} = \frac{1}{\sqrt{2}} (Q_M + i P_M),\quad
    a^*_{M} = \frac{L_{M,-}}{\sqrt{M}} = \frac{1}{\sqrt{2}} (Q_M - i P_M)
  \end{equation}
  with
  \begin{equation}
    L_{M,\pm} = L_{M,1} \pm i L_{M,2} =  \frac{1}{2}\sum_{j=1}^M\sigma_\pm^{(j)}
  \end{equation}
  with $\sigma_\pm = \sigma_1\pm i\sigma_2$ in terms of Pauli
  matrices. $a_M, a_M^*$ are related to the \emph{functions} $a_M(x)$,
  $a_M^*(x)$ by
  \begin{equation} \label{eq:12}
    a_M = \bigoplus_j a_M(x_j) \otimes \Bbb{1},\quad a_M^* = \bigoplus_j
    a_M^*(x) \otimes \Bbb{1}. 
  \end{equation}
  where $x_j$ is given as in Equation (\ref{eq:11}) by $x_j = 2j/M$.
\item 
  The quadratic Hamiltonian $H$ will be written as
  \begin{equation} \label{eq:15}
    H=\sum_{k=0}^3 c_k A_k\ \mbox{with}\ A_0=a^2,A_1=aa^*, A_2=a^*a, A_3=a^{*2}
  \end{equation}
  and the finite dimensional version $\mathcal{H}_M$ becomes similarly
  \begin{equation}
    H_M=\sum_{k=0}^3 c_k A_{M,k}\ \mbox{with}\ A_{M,0}=a_M^2,\ A_{M,1}=a_Ma_M^*,\ A_{M,2}=a_M^*a_M,\ A_{M,3}=a_M^{*2}.
  \end{equation}
  In analogy to Equation (\ref{eq:12}) we can rewrite $H_M$ as a direct sum
  \begin{equation}
    H_M = \bigoplus_j H_M(x_j) \otimes \Bbb{1}
  \end{equation}
  where
  \begin{equation} \label{eq:19}
    H_M(x)=\sum_{j=0}^3c_j B_{M,j}(x) A_j = \sum_{j=0}^3 c_j A_{M,j}(x),
  \end{equation}
  with 
  \begin{equation}
    \begin{array}{ll}
    B_{M,0}(x)=\beta_M(x,{\mathbf {\mathbf N}})\beta_M(x,{\mathbf N}+\Bbb{1}) 
    &B_{M,1}(x)=\beta_M(x,{\mathbf N})^2\\\\
    B_{M,2}(x)=\beta_M(x,{\mathbf N}-\Bbb{1})^2& B_{M,3}(x)=\beta_M(x,{\mathbf
      N}-\Bbb{1})\beta_M(x,{\mathbf N}-2\Bbb{1}) 
  \end{array}
  \end{equation}
\item 
  Finally we can write the unitary $U_{M,t}$ from Equation (\ref{eq:6}) in
  exactly the same way
  \begin{equation}
    U_{M,t} = \bigoplus_j U_{M,t}(x) \otimes \Bbb{1}
  \end{equation}
  with
  \begin{equation} \label{eq:20}
    U_{M,t}(x) = \exp\bigl(-i t H_M(x)\bigr).
  \end{equation}
\item 
  Furthermore we define the multiindex $R\in\{-1,1\}^{d}$, which enables the compact
  notation
  \begin{equation}
      a^R\equiv a_M^{(R_1)}a_M^{(R_2)}\dots a_M^{(R_d)} \quad a_M^{(1)}=a_M^* \quad
      a_M^{(-1)}=a_M
  \end{equation}
  and similarly for the quantities $a^R\in {\mathcal S}({\mathcal
    H}_{\infty})$. In addition we define $|R|\equiv d$ and $w(R)\equiv\sum_i^d
  R_i$.
\item 
  To enable compact notations in some proofs we will write objects describing
  the continuous system sometimes with the subscript $\infty$, i.e.
  \begin{equation} \label{eq:31}
    a_\infty(x) = \sqrt{x}a,\quad a_\infty^*(x) = \sqrt{x}a^*,\quad
    H_\infty(x) = x H,\quad U_{\infty,t}(x) = U_{x,t}, 
  \end{equation}
  where $x$ is as usual in the interval $[0,1]$. 
\end{itemize}

To prove Theorem \ref{thm:1} we have to show that for each $\epsilon > 0$ we
can find an $M_\epsilon$ such that
\begin{equation} \label{eq:14}
  \bigl\lvert \tr(U_{M,t} \rho_M U_{M,t}^* a_M^R) - 
    \tr(U_{\lambda,t} \rho_\infty U_{\lambda,t}^* a^R ) \bigr\rvert < \epsilon
\end{equation}
holds for all $M > M_\epsilon$. Note that we have replaced the polynomials
$f(Q_M,P_M)$ and $f(Q,P)$ from Theorem \ref{thm:1} by monomials $a_M^R$ and
$a^R$ with an arbitrary multiindex $R$. The general statement easily follows
by linearity. To get such an estimate note first that we can rewrite traces of
the form $\tr(U_{M,t} \rho_M U_{M,t}^* a_M^R)$ as
\begin{align}
  \tr(U_{M,t} \rho_M U_{M,t}^* a_M^R) &= \sum_j w_M(j) \tr\bigl(U_{M,t}(x_j)
  R_M(x_j) U_{M,t}^*(x_j) a_M(x_j)\bigr) \\
  &= \int_0^1 \tr\bigl(U_{M,t}(x) R_M(x) U_{M,t}^*(x) a_M(x)\bigr) \mu_M(dx)
\end{align}
where $R_M$ is an integral representation of $\rho_M$ and $\mu_M$ is the
measure define in Equation (\ref{eq:13}). With this we can rewrite the right
hand side of (\ref{eq:14}) as
\begin{multline}
  \bigl\lvert \tr(U_{M,t} \rho_M U_{M,t}^* a_M^R) - 
    \tr(U_{\lambda,t} \rho_\infty U_{\lambda,t}^* a_\infty^R(\lambda) ) \bigr\rvert < \\
    \int_0^1 \bigl\lvert \tr\bigl(U_{M,t}(x) R_M(x) U_{M,t}^*(x) (a_M^R(x) -
    a_\infty^R(x))\bigr) \bigr\rvert \mu_M(dx) + \\ \int_0^1 \bigl\lvert
    \tr\bigl((U_{\lambda,t} \rho_\infty U_{\lambda,t}^* - U_{M,t}(x)
    \rho_\infty U_{M,t}^*(x)) a^R_\infty(x)\bigr)\bigr\rvert \mu_M(dx) \label{eq:9}
\end{multline}
Note that we have used here $a_\infty^*(\lambda)$ and $a_\infty(\lambda)$ from
Equation (\ref{eq:31}). They are creation and annihilation operator belonging to
$Q_\infty$ and $P_\infty$ from Theorem \ref{thm:2} and \ref{thm:1}. 

The discussion of the two integrals on the right hand side is now broken up
into four steps.

\begin{itemize}
\item 
  In Subsection \ref{sec:more-schw-oper} we will collect additional material
  about Schwartz operators, which is used throughout this section.
\item 
  Then we show that the sequences $U_{M,t}^* a_M^R(x) U_{M,t}
  \psi_n$ converge for any Hermite function to $U_{x,t}^* a^R
  U_{x,t} \psi_n$; cf. Subsection \ref{sec:herm-funct-analyt}.
\item 
  In Section \ref{sec:key-estimate} we look at operators $(N + p)^{p/2}
  U_{M,t}^* a_M^R(x) U_{M,t}$ and show that $p \in \Bbb{N}$ can be chosen
  such that the sequence of their norms is uniformly bounded in $x$. This
  allows us to trace convergence of unbounded operators back to convergence of
  bounded operators; cf. Subsection \ref{sec:key-estimate}.
\item 
  The two previous steps together allows us to bound the integrands of the
  integrals in Equation (\ref{eq:9}) and this leads to the required estimate;
  cf. Subsection \ref{sec:proof-main-theorem}
\end{itemize}

\subsection{More on Schwartz operators}
\label{sec:more-schw-oper}

Let us collect some properties of Schwartz operators which are needed for the
proof. The standard reference is \cite{SOP}. The first statement says that
the space $\mathcal{S}(\mathcal{H}_\infty)$ is stable under multiplication with
polynomials in $Q$ and $P$. 

\begin{prop} \label{prop:5}
  Consider a polynomial $f(Q,P)$ in position $Q$ and momentum $P$ and a
  Schwartz operator $\rho$. The products $f(Q,P) \rho$ and $\rho f(Q,P)$ are
  closeable operators and their closures are again Schwartz
  operators\footnote{In slight abuse of language we will shorten in the future
  statements of this form as: ``$\dots$ is a Schwartz operator''.}. The two
  maps 
  \begin{gather}
    \mathcal{S}(\mathcal{H}_\infty) \ni \rho \mapsto f(Q,P) \rho \in
    \mathcal{S}(\mathcal{H}_\infty)\\
    \mathcal{S}(\mathcal{H}_\infty) \ni \rho \mapsto \rho f(Q,P) \in
    \mathcal{S}(\mathcal{H}_\infty) 
  \end{gather}
  defined that way are continuous.
\end{prop}

\begin{proof}
  This is an easy consequence of the definition.
\end{proof}

Furthermore the space $\mathcal{S}(\mathcal{H}_\infty)$ is closed under
multiplication. In other words \cite{SOP}

\begin{prop} \label{prop:4}
  The product of two Schwartz operators is again a Schwartz operator.
\end{prop}

The next result is an alternative characterization of Schwartz operators. To
this end let us first introduce for any $x = (x_1,x_2) \in \Bbb{R}^2$ the Weyl
operator acting on $\psi \in \mathrm{L}^2(\Bbb{R})$ by
\begin{equation}
  (W(x_1,x_2)\psi)(x) = e^{-i x_1 x_2/2}e^{i x_2 x }\psi(x-x_1).
\end{equation}
The $W$ are unitary and for any density operators $\rho$ they give rise to the
\emph{characteristic function} 
\begin{equation}
  \Bbb{R}^2 \ni x \mapsto \hat{\rho}(x) = \tr(\rho W(x)) \in \Bbb{C}
\end{equation}
Now we have \cite{SOP}:

\begin{prop} \label{prop:2}
  A density operator $\rho$ is a Schwartz operator iff its characteristic
  function $\hat{\rho}$ is a Schwartz function.
\end{prop}

Now let us come back to the quadratic Hamiltonians $H$ and $H_M(x)$ defined in
Equation (\ref{eq:15}) and (\ref{eq:19}), and the corresponding time
evolutions $U_{\lambda,t}$ and $U_{M,t}(x)$ from (\ref{eq:8}) and (\ref{eq:20}).

\begin{lem} \label{lem:1}
  For each Schwartz operator $\rho \in \mathcal{S}(\mathcal{H}_\infty)$ and each
  multiindex $R \in \{-1,1\}^d$ the following statements hold:
  \begin{enumerate}
  \item \label{item:5}
    $U_{\lambda,t} \rho U_{\lambda,t}^*$ and $U_{M,t}(x) \rho U_{M,t}^*(x)$ are
    Schwartz operators.
  \item \label{item:6}
    $U_{\lambda,t}^* a^R U_{\lambda,t} \rho$ and $U_{M,t}^*(x) a^R U_{M,t}(x)
    \rho$ are Schwartz operators (and in particular trace class).
  \item \label{item:7}
    We have
    \begin{align}
      \tr\bigl(U_{\lambda,t}^* a^R U_{\lambda,t} \rho\bigr) &= \tr\bigl( a^R
      U_{\lambda,t} \rho U_{\lambda,t}^*\bigl)\\
      \tr\bigl(U_{M,t}^*(x) a^R U_{M,t}(x) \rho\bigr) &= \tr\bigl( a^R
      U_{M,t}(x) \rho U_{M,t}^*(x)\bigr) .
    \end{align}
  \end{enumerate}
\end{lem}

\begin{proof}
  To prove the first statement in \ref{item:5} note first that for quadratic
  Hamiltonians the canonical operators $Q_t = U_{t,\lambda}^* Q U_{t,\lambda}$
  and $P_t = U_{t,\lambda}^* P U_{t,\lambda}$ evolve according to the
  corresponding classical equation of motion (this is well known or 
  otherwise an easy exercise). This implies for the Weyl operator that
  $U_{t,\lambda}^* W(x) U_{t,\lambda} = W(F_t x)$ holds with a linear map
  $F_t$ (the classical phase space flow). Hence we get 
  \begin{equation}
    \widehat{(U_{\lambda,t} \rho U_{\lambda,t}^*)}
  = \widehat{\rho} \circ F_t,
  \end{equation}
  which implies that $\widehat{(U_{\lambda,t} \rho U_{\lambda,t}^*)}$
  is a
  Schwartz function and therefore $U_{\lambda,t} \rho U_{\lambda,t}^*$ is a
  Schwartz operator by Proposition \ref{prop:2}. 

  Now consider the unitary $U_{M,t}(x)$. Since its generator $H_M(x)$
  satisfies $E_{j_x} H_M(x) E_{j_x} = H_M(x)$ with $j_x = \frac{\lfloor M  x
    \rfloor}{2}$ from Equation (\ref{eq:11}) and the projection $E_{j_x}$ 
  defined in (\ref{eq:21}), we can rewrite $U_{M,t}(x)$ as
  \begin{equation}
    U_{M,t}(x) = (\Bbb{1} - E_{j_x}) + \tilde{U}_{M,t}(x), \quad
    \tilde{U}_{M,t}(x) = E_{j_x} U_{M,t}(x) E_{j_x}.
  \end{equation}
  Hence 
  \begin{multline} \label{eq:22}
    U_{M,t}(x) \rho U_{M,t}^*(x) = \\ (\Bbb{1} - E_{j_x}) \rho (\Bbb{1} -
    E_{j_x}) + (\Bbb{1} - E_{j_x}) \rho \tilde{U}_{M,t}^*(x) +
    \tilde{U}_{M,t}(x) \rho (\Bbb{1} - E_{j_x}) + \tilde{U}_{M,t} \rho
    \tilde{U}_{M,t}^*. 
  \end{multline}
  Since $\rho \in \mathcal{S}(\mathcal{H}_\infty)$ by assumption, Proposition
  \ref{prop:3} immediately implies that $(\Bbb{1} - E_{j_x}) \rho (\Bbb{1} -
  E_{j_x})$, $(\Bbb{1} - E_{j_x}) \rho$ and $\rho (\Bbb{1} - E_{j_x})$ are
  all Schwartz operators. The same is true for $\tilde{U}_{M,t}$ and
  $\tilde{U}_{M,t}^*$ (again by Proposition \ref{prop:3}). Hence by
  Proposition \ref{prop:4} all operators on the right hand side of
  (\ref{eq:22}) are Schwartz operators and therefore $U_{M,t}(x) \rho
  U_{M,t}^*(x)$ is a Schwartz operator as well.

  To show statement \ref{item:6} note first that $a^R U_{t,\lambda} \rho
  U_{t,\lambda}^*$ is a Schwartz operator -- this follows from Proposition
  \ref{prop:5} and statement \ref{item:5} above. Applying item \ref{item:5}
  again we see that
  \begin{equation}
    U_{\lambda,t}^* a^R U_{\lambda,t} \rho = U_{\lambda,t}^* (a^R U_{t,\lambda} \rho
    U_{t,\lambda}^*) U_{t,\lambda} 
  \end{equation}
  is a Schwartz operator as well. Similar reasoning holds for $U_{M,t}^*(x)
  a^R U_{M,t}(x) \rho$. 

  The last statement is an easy consequence of item \ref{item:6} and elementary
  properties of the trace.
 \end{proof}

\subsection{Hermite functions and analytic vectors for quadratic Hamiltonians}
\label{sec:herm-funct-analyt}

The subject of this section is convergence of sequences $a^S U_{M,t}(x)
\psi_n$, $M \in {\mathbb N}$ to $a^S U_{x,t} \psi_n$. This will lead to a first partial
result in Proposition \ref{prop:6} stating that the $U_{M,t}(x)$ converge
strongly to $U_{x,t}$. The main technical tool will be the fact that quadratic
Hamiltonians admit (finite) linear combinations of Hermite functions as
analytic vectors (cf. \cite[Ch. 5]{Folland}). The first step is a general
estimate. 

\begin{lem}
  \label{csek} Consider a Hermite function $\psi_n$ and a multiindex
  $S\in\{-1,1\}^{2d}$.  Then the following bound holds for all $M \in \Bbb{N}
  \cup \{\infty\}$ and $x \in [0,1]$. 
  \begin{equation}
    \|a^S H_M(x) \psi_n\|\leq 2^{3d+\frac{n}{2}}d!m!(32 \max
    |c_j|)^m\label{18}
  \end{equation}
  with the maximum $\max|c_j|$ taken over the $j=0,\dots,3$.
\end{lem}

\begin{proof}
  It is sufficient to show Equation (\ref{18}) for the Hamiltonian $H =
  H_\infty(1)$ instead of $H_M(x)$ since all bounds hold also when we
  substitute $a_M(x)^R$ for $a^R$. The former differs from the latter only by
  a prefactor between $0$ and $1$. Now consider the bound 
  \begin{equation} 
    [\|a^R\psi_n\|\leq 2^{\frac{n}{2}}4^d d!\ .\label{bound} 
  \end{equation}
  \begin{equation}
    \|a^R\psi_n\|\leq \sqrt{(n+1)(n+2)\cdots
      (n+2d)}\leq\sqrt{2^{n+2d}}\sqrt{(2d!)}\leq 2^{\frac{n}{2}+d}2^d
    d!=2^{\frac{n}{2}}4^d d!\ ,
  \end{equation} 
  where to see the first inequality we recall
  \begin{equation}
    a^*\psi_n=\sqrt{n+1}\psi_{n+1},\;a\psi_n=\sqrt{n}\psi_{n-1}\ ,
  \end{equation} 
  from which the inequality is clear as it is an equality exactly in the worst
  case of $a^R=(a^*)^{2d}$. For the next two inequalities we used
  \begin{equation}
    \frac{(p+q)!}{p!q!}\leq 2^{p+q}\ ,\label{bintr}
  \end{equation}
  (which holds since the lhs. is a term in the binomial expansion of the rhs.)
  with $p=n$, $q=2d$ and $p=q=d$, respectively.  

  Let us now introduce the multiindex ${\bf j}\in\{0,1,2,3\}^m$ to enable the
  compact notation
  \begin{equation}
    H^m\equiv \sum_{\bf j}c_{\bf j} A_{\bf j}\;\mbox{with}\; x_{\bf
      j}\equiv x_{j_1}x_{j_2}\dots x_{j_m}
  \end{equation}
  Now, taking into account that 
  \begin{equation}
      |c_{\mathbf{j}}| \leq |\max_{j=0,\dots,3} c_j|^m \quad \forall \mathbf{j};
  \end{equation}
  and that $A_{\bf j}=a^R$ for each ${\bf j}$ with an $R$ satisfying $|R|=2m$,
  thus 
  \begin{equation}
    a^S A_{\bf j}=a^{R'}\ \text{with}\ |R'|=2(d+m)
  \end{equation}
  we can calculate
  \begin{align}
    \|a^S H_M(x)^m\psi_n\| &\leq  \sum_{\bf j} |c_{\bf j}| \|a^S A_{\bf j}
    \psi_n\| \\
    & \leq 4^m(\max |c_j|)^m 2^{\frac{n}{2}}4^{d+m}(d+m)!\\
    & \leq (32 \max(|c_j|))^m 2^{3d+\frac{n}{2}}m!d!  \ , 
  \end{align}
  where in the second inequality we inserted the bound (\ref{bound}) for
  $|R'|=2(m+d)$. The third inequality is due to the application of
  (\ref{bintr}) with $p=m, q=d$.
\end{proof}

\begin{lem}
  The bound in (\ref{18}) also holds if $|s|=2d-1$
\end{lem}

\begin{proof}
  We obviously have
  \begin{equation}
    \|a\,a^S H_M(x)^m \psi_n\|\geq \|a^S \tilde{H}^m \psi_n\|\ ,
  \end{equation}
  so the statement follows from Lemma \ref{csek}.
\end{proof}

\begin{lem} \label{lem:2}
  For each $M, n \in \Bbb{N}$ and $x \in [0,1]$ the following
  bound holds
  \begin{equation} 
    |\sqrt{x} - \beta_M(x,n)| \leq \sqrt{\frac{n}{M}}
  \end{equation}
\end{lem}

\begin{proof}
  Recall the definition of $\beta$ (cf. Equation (\ref{eq:23}):
  \begin{equation}
    \beta: \mathbb{N}^2\times [0,1]\to [0,1]: (M,n,x)\mapsto \beta_{M}(x,n) =
    \theta_1(x -\frac{n}{M})\ \text{with}\ \theta_1(x)=\sqrt{x\chi_{[0,1]}(x)}.    
  \end{equation}
  If $x < n/M$ the argument of the root is negative and we
  get $\beta_M(x,n) = 0$, hence 
  \begin{equation}
    |\sqrt{x} - \beta_M(x,n)| = \sqrt{x} \leq \sqrt{\frac{n}{M}}
  \end{equation}
  as stated. If $x \in [n/M,1]$ we can check that the function 
  \begin{equation}
    [n/M,1] \ni x \mapsto \sqrt{x} - \sqrt{x - \frac{n}{M}} \in \Bbb{R}
  \end{equation}
  is positive and monotonically decreasing (the latter can be easily seen by
  looking at the derivative). Hence the biggest value is achieved at the
  beginning of the interval (i.e. $x= n/M$) and we get again the statement of
  the lemma. 
\end{proof}

\begin{lem}
  \label{whichone} For any pair of multiindices $S,R$ and $n\in\mathbb{N}$ we
  have
  \begin{equation}
    \lim_{M\to\infty}\|a^S(a^R_M(x)- x^{|R|/2}a^R)\psi_n\|=0
  \end{equation}
  uniformly in $x$.
\end{lem}

\begin{proof}
  Note first that for $r\in\{-1,1\}$
  \begin{equation} 
    a^{(r)}_M\psi_n-\xpl a^{(r)}\psi_n=\left[\beta_{M}\left(x,n+\frac{r-1}{2}\right)-\xpl\right]
    \sqrt{n+\frac{r+1}{2}}\psi_{n+r}\label{onet}
  \end{equation}
  The general case can be schematically written as
  \begin{equation}
    a^S(a^R_M(x)-x^{|R|/2}a^R)\psi_n\sim\prod_{j=1}^d\sqrt{n+p_j}\left(
      \beta_{M}(x,n+p'_j)-\xpl\right)a^S\psi_{n+q}\label{schem}
  \end{equation}
  with $d=|R|$ and $p_i,p'_i, q\in\mathbb{N}$ can be easily written in terms
  of $w_j(R)\equiv\sum_{i=1}^j r_j$ with $R=(r_d,r_{d-1},\dots,r_1)$. Note,
  that the omitted factor of proportionality is $M$-independent. Now we have
  according to Lemma \ref{lem:2} for any $p\in\mathbb{N}$
  \begin{equation}
    \left\lvert \sqrt{x} - \beta_M(x,n+p)\right\rvert \leq \sqrt{\frac{n+p}{M}}.
  \end{equation}
  Hence $\lim_{M\rightarrow\infty} \beta_M(x,n+p) = \sqrt{x}$ uniformly, which
  proves the statement.
\end{proof}

Now, we can prove convergence of $U_{M,t}(x) \psi_n$ for small times.

\begin{lem} 
  \label{teft}
  Let $x\in [0,1], n\in{\mathbb N}$ and $|t|<(32\max|c_j|)^{-1}$. Then 
  \begin{equation} \label{eq:25}
    \lim_{M\to
      \infty}\|a^SU_{t,M}(x)\psi_n-a^SU_{t,x}\psi_n\|=0
  \end{equation}
  uniformly in $x$.
\end{lem}

\begin{proof}
  Note first that the expressions in Equation (\ref{eq:25}) are well defined,
  since $U_{M,t}(x) \psi_n$ and $U_{x,t}\psi_n$ are Schwartz functions (this
  follows from Lemma \ref{lem:1}) and therefore in the domain of $a^S$. 

  Now proceed with $q\equiv 32|t|\max|c_j|<1$ which implies
  \begin{equation} 
    \|a^S\frac{(itH)^m}{m!}\psi_n\|\leq 2^{3d-\frac{n}{2}}d!q^m\
    .\label{gser}
  \end{equation}
  The sum over $m\in{\mathbb N}$ is then convergent since the rhs. of
  (\ref{gser}) is also (since $q < 1$). Hence we can write $U_{x,t} \psi_n$ in
  terms of the exponential series
  \begin{equation}
    U_{x,t}\psi_n = \exp(-itx H)\psi_n = \sum_{n=0}^\infty
    \frac{(-itxH)^m \psi_n}{m!}. 
  \end{equation}
  Now note that $U_{x,t} \psi_n \in \mathcal{S}(\Bbb{R})$ and therefore $a^S
  U_{x,t} \psi_n$ is well defined. Hence we get
  \begin{equation}
    a^S U_{x,t}\psi_n = a^S \exp(-itx H)\psi_n = \sum_{n=0}^\infty 
    \frac{a^S (-itxH)^m \psi_n}{m!}.
  \end{equation}
  The last equality holds (although $a^S$ is unbounded) because $a^S$ is
  closable and convergence of the series is known from (\ref{gser}).
  
  Now, consider
  \begin{multline} \label{eq:24}
    \|a^S (U_{M,t}(x) - U_{x,t}) \psi_n\| = \| \exp(itH_M(x))-\exp(itx
    H)]\psi_n\| \leq \\
    \left\|\sum_{m>\mu}\frac{a^S(itH_M(x))^m}{m!}\psi_n\right\|
    + \left\|\sum_{m>\mu}\frac{a^S(itxH)^m}{m!}\psi_n\right\| + \\
    \left\|\sum_{m\leq
        \mu}a^S\frac{(itH_M(x))^m-(itxH)^m}{m!}\psi_n\right\|
  \end{multline}
  Now, for all $\epsilon>0$ there is $\mu_\epsilon\in\mathbb{N}$
  (independent of $M$ and $x$) such that both the first two terms are smaller
  than $\epsilon/3$, since they are bounded by the remainder of the geometric
  series c.f.  (\ref{gser}) and the remark after it. For the third term we
  write
  \begin{align}
    \left\|\sum_{m\leq\mu_\epsilon}a^S\frac{(itH_M(x))^m-(itx
        H)^m}{m!}\psi_n\right\| & \leq
    \sum_{m\leq\mu_\epsilon}\left\|a^S\left(H_M(x)^m-x^m
        H^m\right)\psi_n\right\|\frac{|t|^m}{m!}\\
    & \leq \sum_{m\leq\mu_\epsilon} \sum_{R\in\{-1,1\}^{2m}}\left\|a^S\left(a^R_M(x)-x^{|R|/2}
        a^R\right)\psi_n\right\|\frac{|t|^m}{m!}\ . 
  \end{align}
  The last expression contains finite number of terms, each of which converges
  (uniformly in $x$) to zero when $M\to \infty$ due to Lemma (\ref{whichone}),
  hence for $M$ large enough the last term in (\ref{eq:24}) is smaller
  than $\epsilon/3$ as well.
\end{proof}

This lemma leads to the following strong convergence result, which is a first
step towards Theorem \ref{thm:1}.

\begin{prop} \label{prop:6}
  \label{strc} For any fixed $\eta \in \mathcal{H}_\infty$ 
  \begin{equation}
    \lim_{M \rightarrow \infty} \| (U_{M,t}(x) - U_{x,t})\eta \| = 0
  \end{equation}
  holds uniformly in $x$.
\end{prop}

\begin{proof}
  Consider $\eta\in{\mathcal H}_\infty$ with $\|\eta\|=1$ and its expansion in
  terms of $\psi_n$: $\eta=\sum_n\eta^{(n)}\psi_n$. For any $\epsilon>0$ we
  can decompose $\eta=\eta_1+\eta_2$ such that
  \begin{equation}
    \eta_1=\sum_{n\leq n_\epsilon}\eta^{(n)}\psi_n,\quad
    \eta_2=\sum_{n>n_\epsilon}\eta^{(n)}\psi_n,
    \quad\|\eta_2\|<\epsilon\label{stct}
  \end{equation}
  Now we write
  \begin{align}
    \|(U_{M,t}(x)-U_{x,t})\eta\|
    &\leq\|(U_{M,t}(x)-U_{x,t})\eta_1\|+\|(U_{M,t}(x)-U_{x,t})\eta_2\| \\ 
    &\leq\sum_{n\leq n_\epsilon}
    \eta^{(n)}\|(U_{M,t}(x)-U_{x,t})\psi_n\|+2\epsilon\label{maruk} 
  \end{align}
  According to Lemma (\ref{teft}) for each $n\leq n_\epsilon$ there is an
  $M_{n,\epsilon}$ such that for $M>M_{n,\epsilon}$ and all $x$ we have
  $\|(U_{M,t}(x)-U_{x,t})\psi_n\|\leq \epsilon/n_\epsilon$ Since $\|\eta\|=1$ we get
  \begin{equation}
    \|(U_{M,t}(x)-U_{x,t})\eta_1\|\leq
    n_\epsilon\frac{\epsilon}{n_\epsilon}=\epsilon\quad\forall 
    M>M_\epsilon\equiv\max_{n\leq n_\epsilon} M_{n,\epsilon}\ .
  \end{equation}
  Together with (\ref{maruk}) this leads to $\lim_{M\to \infty}
  U_{M,t}(x)\eta=U_{x,t}\eta$. Since $\eta\in{\mathcal H}_\infty$ was arbitrary, the
  statement follows. 
\end{proof}

Now we are aiming at convergence (in a sense we will make precise later) of
$U_{M,t}(x) \rho U_{M,t}(x)^*$ to $U_{x,t} \rho U_{x,t}^*$. The next lemma is
the first step.

\begin{lem}
  \label{ruc}$\forall n\in\mathbb{N}$ we have
  \begin{equation} \label{eq:26}
    \lim_{M\to\infty} U_{M,t}(x) a^S U^*_{M,t}(x)\psi_n=U_{x,t} a^S U^*_{x,t}\psi_n
  \end{equation}
  uniformly in $x$.
\end{lem}

\begin{proof}
  As in Lemma \ref{teft} note first that the expression in (\ref{eq:26}) is
  well defined since $U^*_{M,t}(x) \psi_n$ and $U_{x,t}\psi_n$ are in the
  domain of $a^S$. 

  From the same lemma we have in addition that $\Phi_{M,n}=a^S U^*_{M,t}(x)\psi_n$
  converges uniformly to $\Phi_n = a^S U^*_{x,t}\psi_n$. Hence
  \begin{equation}
    \|U_{M,t}(x)\Phi_{M,n}-U_{x,t}\Phi_n\| \leq
    \|U_{M,t} (\Phi_{M,n} - \Phi_n)\| + \|(U_{M,t} - U_{x,t}) \Phi_n \| < 2 \epsilon
  \end{equation}
  if $M>\max\{M_\epsilon,M'_\epsilon\}$ and for arbitrary $x$, where
  $M_\epsilon$ comes from the strong convergence of the second term proved in
  Prop. \ref{strc} whereas $M'_\epsilon$ from the above mentioned convergence.
\end{proof}

\subsection{The key estimate}
\label{sec:key-estimate}

The purpose of this Section is to prove the following Lemma which will allow
us to trace the convergence of sequences of unbounded operators $U_{M,t}^*(x)
a^S U_{M,t}(x)$ back to convergence of bounded operators.

\begin{lem}
  \label{yab}$\exists p, p' \in \mathbb{N} \cup \{\infty\}$ and $K\in\mathbb{R}_+$
  such that there is a bounds
  \begin{equation} \label{eq:34}
    \|(\no+p\Bbb{1})^{-\frac{p}{2}} U_{M,t}(x) a^S U^*_{M,t}(x)\|<K
  \end{equation}
  and
  \begin{equation} \label{eq:35}
    \| U_{M,t}(x) a^S U^*_{M,t}(x) (\no+p'\Bbb{1})^{-\frac{p'}{2}} \|<K'
  \end{equation}
  for sufficiently small $t$ and $\forall M\in {\mathbb N}$ and for all $x \in
  [0,1]$. 
\end{lem}

\begin{proof}
  Since $\psi_n$, $n\in \mathbb{N}$ is a basis, it is sufficient to show that
  \begin{equation}
    \|(\no+p\Bbb{1})^{-\frac{p}{2}} U_{M,t}(x) a^S U^*_{M,t}(x)\psi_n\|<K\
    \label{nli}
  \end{equation}
  and
  \begin{equation} \label{eq:33}
    \| U_{M,t}(x) a^S U^*_{M,t}(x) (\no+p\Bbb{1})^{-\frac{p}{2}} \psi_n\|<K\
  \end{equation}
  hold. Furthermore, it is sufficient to show this for $n>n_0$ since the overall
  bound will then be given by $\max\{\tilde{K},K_n|n<n_0\}$ where we denoted
  $\tilde{K}$ the bound for $n>n_0$, and $K_n$ is a bound for fixed $n$'s in
  (\ref{nli}), which exists due to Lemma \ref{ruc}. For the proof of
  (\ref{nli}) and (\ref{eq:33}) $\forall n>n_0$ with a suitable
  $n_0\in\mathbb{N}$ we proceed by introducing the spectral projections of the
  number operator $E_{(a,b)}=\chi_{[a,b]}(\no)$ and  proving some technical lemmas.

  \begin{lem}
    \label{tech1}
    There is an $n_\epsilon\in\mathbb{N}$ and $t_\epsilon\in\mathbb{R}_+$ such
    that
    \begin{equation}
      \|(\Bbb{1}-\spe)\Phi_{M,n}\|<\epsilon \ \text{with} \ \Phi_{M,n}=a^S
      U^*_{M,t}(x)\psi_n 
    \end{equation}
    $\forall M\in\Bbb{N},\;|t|<t_\epsilon,\;n>n_\epsilon$.
  \end{lem}

  \begin{proof}
    Since $a^S H^m_M(x)$ is a polynomial in $a_M(x),a_M^*(x)$ of order not greater than
    $d+2m$ (with $d=|S|$), we have
    \begin{equation}
      (\Bbb{1}-\spe)\Phi_{M,n}^{(\mu)}=0\quad\mbox{if}\quad\frac{n}{2}\geq
      d+2\mu\ ,
    \end{equation}
    where $\Phi_{M,n}^{(\mu)}$ denote the partial sums
    \begin{equation}
      \Phi_{M,n}^{(\mu)} = \sum_{m<\mu}\frac{a^S\bigl(itH_M(x)\bigr)^m}{m!}\psi_n.
    \end{equation}
    Hence for the limit of the series the following formula holds:
    \begin{equation}
      (\Bbb{1}-\spe)\Phi_{M,n} = ((\Bbb{1}-\spe)(\Phi_{M,n} -
      \Phi_{M,n}^{((n-2d)/4)})\ .
    \end{equation}
    Let us estimate
    \begin{align}
      \|(\Bbb{1}-\spe) (\Phi_{M,n}-\Phi_{M,n}^{((n-2d)/4)})\| & \leq
      \|(\Phi_{M,n} - \Phi_{M,n}^{((n-2d)/4)}) \| \\
      & \leq\sum_{m>\frac{n-2d}{4}} \frac{\|a^S H_M^m(x)\Psi_n\|}{m!}|t|^m \\ 
      & \leq 2^{3d+\frac{n}{2}}d!\!\!\sum_{m>\frac{n-2d}{4}}\!\!q^m
      \label{est}
    \end{align}
    with $q=32|t|\max|c_j|$. With sufficiently small $t$, c.f. Lemma
    \ref{teft}, that is, $q<1$ is assumed, we can sum the geometric series
    explicitly:
    \begin{equation} 
      2^{3d+\frac{n}{2}} d! \frac{q^{\frac{n-2d}{4}+1}}{1-q}
      =\frac{2^{3d}d!}{1-q}\exp(\frac{2-d}{2}\log q)\exp[n(\log
      \sqrt{2}+\frac{1}{4}\log q)]\equiv
      K_{1,q}\,e^{-K_{2,q}n}\label{lvm}
    \end{equation}
    where the constants $K_{i,q}$ can be read off from the formula.  Now
    shortening the time interval such that $q<1/4$ we find make $K_{2,q}<0$.
    Thus, we can conclude that there is an $n_{\epsilon,q}$ such that the
    statement of the Lemma is fulfilled.
  \end{proof}

  \begin{lem}
    \label{tech2}
    \begin{equation}
      \|\spe\Phi_{M,n}\|\leq\left(\frac{3n}{2}+2d\right)^{\frac{d}{2}}
    \end{equation}
  \end{lem}

  \begin{proof}
    Let us expand $\Phi_{M,n}$ and $U_{M,n}(x) \psi_n$ in terms of Hermite
    functions 
    \begin{equation}
      \Phi_{M,n}=\sum_k C_k\psi_k,\quad \phi=\sum_k c_k\psi_k\ .
    \end{equation}
    Then $C_{k+w(S)}=\Lambda(k,S)c_k$ holds with $\Lambda(k,S)$ satisfying
    $(k-d)^{d/2}<\Lambda(k,S)< (k+d)^{d/2}$, where we recall $|S|=d$. Since the
    relation $-d\leq w(S)\leq d$ also holds, the condition $k+w(S)\in
    [n/2,3n/2]$ implies $k\in[n/2-d,3n/2+d]$. From these we get
    \begin{equation}
      \|\spe \Phi_{M,n}\|^2\leq\sum_j
      |c_{j-w(S)}|^2(\frac{3n}{2}+2d)^d
    \end{equation}
    and since $\sum_k|c_k|^2=1$ due to $\|\psi_n\|=1$ and the unitarity of
    $U_{M,t}(x)$ we finally get to the statement of the Lemma.
  \end{proof}

  Now let us return to the bound in Equation (\ref{nli}) and decompose the
  norm in the following way. 
  \begin{align}
    \|(\no+p\Bbb{1})^{-\frac{p}{2}}U_{M,t} \Phi_{M,n}\| &\leq
    \|(\no+p\Bbb{1})^{-\frac{p}{2}}\speb U_{M,t} \spe\Phi_{M,n}\| \nonumber \\ 
    &+ \|(\no+p\Bbb{1})^{-\frac{p}{2}}(\Bbb{1}-\speb) U_{M,t} \spe\Phi_{M,n}\|
    \nonumber \\ 
    &+ \|(\no+p\Bbb{1})^{-\frac{p}{2}}U_{M,t} (\Bbb{1}-\spe)\Phi_{M,n}\|
  \end{align}
  The third term can be estimated as
  \begin{equation}\|(\no+p\Bbb{1})^{-\frac{p}{2}}U_{M,t} (\Bbb{1}-\spe)\Phi_{M,n}\|\leq 
    \|(\no+p\Bbb{1})^{-\frac{p}{2}}\|\|(\Bbb{1}-\spe)\Phi_{M,n}\|\leq p^{-\frac{p}{2}}\epsilon\end{equation}
  whenever $n> n_\epsilon$ due to Lemma \ref{tech1}. 
  For the second term we write
  \begin{multline}
    \|(\no+p\Bbb{1})^{-\frac{p}{2}}(\Bbb{1}-\speb) U_{M,t} \spe\Phi_{M,n}\|
    \leq \\ 
    \|(\no+p\Bbb{1})^{-\frac{p}{2}}\|\max_{k\in [\frac{n}{2},\frac{3n}{2}]}
    \|(\Bbb{1}-\speb) U_{M,t}\psi_k\|\|\spe\Phi_{M,n}\|\leq \\
    p^{-\frac{p}{2}} K_{1,q} e^{-K_{2,q}\frac{n}{2}} \left(\frac{3n}{2} +
      2d\right)^\frac{d}{2} \ ,
  \end{multline}
  where we used the Lemma \ref{tech2} and the technique of the proof of Lemma \ref{tech1} 
  for the middle and right norms in the second line
  to arrive at the third. Note, that for $k\in [n/2,3n/2]$ we have that $[k/2,3k/2]\supset [n/4,9n/4]$,
  which justifies using the same constants $K_{i,q}$ as in Lemma \ref{tech1}.    
  It is now clear, that increasing $n$ makes this bound arbitrary small.
  Finally, note that $(\no+p\Bbb{1})^{-p/2}$ commutes with $E_{[n/4,9n/4]}$, and since the latter is
  a projection we can write for the first term
  \begin{multline}
    \|(\no+p\Bbb{1})^{-\frac{p}{2}}\speb U_{M,t} \spe\Phi_{M,n}\| \leq \\ 
    \|\speb(\no+p\Bbb{1})^{-\frac{p}{2}}\speb\|\|U_{M,t}\| \| \spe \Phi_{M,n}\|\
    ,
  \end{multline}
  Now, for the norms we have
  \begin{equation} \label{eq:36}
    \|\speb(\no+p\Bbb{1})^{-\frac{p}{2}}\speb\|\leq 2^p (n+4p)^{-\frac{p}{2}},\quad 
    \|\spe \Phi_{M,n}\|\leq\left(\frac{3n}{2}+2d\right)^{\frac{d}{2}}\ ,
  \end{equation}
  the multiplication of which, after substitution $p=d$, gives $6^d r^{d/2}<6^d$ since 
  $r\equiv (n+4d/3)/(n+4d)<1$. Hence we have found $\tilde{K}=6^d+2\epsilon$,
  which leads to (\ref{nli}) and therefore (\ref{eq:34}). To get (\ref{eq:33})
  and (\ref{eq:35}) we can proceed in exactly the same way. The only
  difference is that the operator $(\no+p'\Bbb{1})^{-\frac{p'}{2}}$ is acting
  directly on $\psi_n$ and not at $U_{M,t} \spe\Phi_{M,n}$. Therefore we only
  have to replace the first norm in (\ref{eq:36}) by $(n+p')^{p'/2}$ to get
  the desired result and the proof is complete.
\end{proof}

Now we are ready to prove a convergence statement for the time evolution
$U_{M,t}(x) \rho U_{M,t}^*(x)$ which brings a step closer to our goal.

\begin{prop} 
  \label{Ss}For any Schwartz operator $\rho\in\cal{S}(\cal{H}_\infty)$ we have 
  \begin{equation} \label{eq:27}
   \lim_{M\to \infty}\tr [a^S U_{M,t}(x)\rho U_{M,t}^*(x)]=\tr [a^S
   U_{x,t}\rho U_{x,t}^*]
  \end{equation}
  uniformly in $x$.
\end{prop}

\begin{proof}
  According to Lemma \ref{lem:1} Equation (\ref{eq:27}) is equivalent to
  \begin{equation}
    \lim_{M\to \infty}\tr [U_{M,t}^* a^S U_{M,t}(x)\rho]=\tr [U_{x,t}^* a^S
   U_{x,t}\rho]
  \end{equation}
  To prove this let us introduce the notations 
  \begin{equation} 
    X_M(x)\equiv(\no+p\Bbb{1})^{-\frac{d}{2}}U_{M,t}^*(x)\,a^S\,U_{M,t}(x)\,,\quad
    X(x)\equiv(\no+p\Bbb{1})^{-\frac{d}{2}}U_{x,t}^*\,a^S\,U_{x,t}\ ,
  \end{equation}
  and show first strong convergence of the defined operators.

  \begin{lem}
    For any $\eta \in \mathcal{H}_\infty$ we have
    \begin{equation}
      \lim_{M \rightarrow \infty} \| (X_M(x) - X(x)) \eta \| = 0
    \end{equation}
    uniformly in $x$.
  \end{lem}

  \begin{proof}
    We use the same strategy as for Proposition \ref{strc} and decompose
    $\eta\in\hi_\infty$ as 
      \begin{equation}
    \eta_1=\sum_{n\leq n_\epsilon}\eta^{(n)}\psi_n,\quad
    \eta_2=\sum_{n>n_\epsilon}\eta^{(n)}\psi_n,
    \quad\|\eta_2\|<\epsilon \; ,
  \end{equation}
  cf. formula (\ref{stct}). Then
    \begin{align}
      \|(X_M(x)-X(x))\eta\| &\leq\|(X_M(x)-X(x))\eta_1\|+\|(X_M(x)-X(x))\eta_2\| \\ 
      &\leq\sum_{n\leq n_\epsilon}|\eta^{(n)}|\|(X_M(x)-X(x))\psi_n\|+2K\epsilon\\ 
      &\leq\big(\sum_{n\leq
        n_\epsilon}|\eta^{(n)}|\big)\,\|(\no+d\Bbb{1})^{-\frac{d}{2}}\| \frac{\epsilon}{n_\epsilon}
      + 2K\epsilon\ ,
    \end{align}
    where the second inequality is the consequence of Lemma \ref{yab}, whereas
    for the third one assumes that $M>\max_{n<n_\epsilon}(M_\epsilon)$ in the
    statement of Lemma \ref{ruc}. Since $\big(\sum_{n\leq
      n_\epsilon}|\eta^{(n)}|\big)\leq n_\epsilon$ and
    $(\no+d\Bbb{1})^{-\frac{d}{2}}$ is bounded, the convergence is proved.
  \end{proof}

  The sequence $X_M$ converges strongly and is norm bounded. This implies for
  any trace class operator convergence of traces $\tr(X_M(x) \rho)$ with
  $\lim_{M \rightarrow \infty} \tr(X_M(x) \rho) = \tr(X(x) \rho)$. Hence we get 
  \begin{equation}
    \lim_{M \rightarrow\infty} \bigl\lvert \tr\bigl(X_M(x) \rho\bigr) -
    \tr\bigl(X(x) \rho\bigr)\bigr\rvert = 0
  \end{equation}
  uniformly in $x$. Therefore we get with a \emph{Schwartz operator} $\rho$
  \begin{align}
    \lim_{M\to\infty}\tr(U_{M,t}^*(x)\,a^S\,U_{M,t}(x)\,\rho) &=
    \lim_{M\to\infty}\tr \big((\no+p\Bbb{1})^{\frac{d}{2}}(\no+p\Bbb{1})^{-\frac{d}{2}}
    U_{M,t}^*(x)\,a^S\,U_{M,t}(x)\,\rho\big) \\
    &=\lim_{M\to\infty}\tr \big(X_M\rho (\no+p\Bbb{1})^{\frac{d}{2}}\big)\\ 
    &= \lim_{M\to\infty}\tr X_M\tilde{\rho}=\tr
    X\tilde{\rho}=\tr(U_{x,t}^*\,a^S\,U_{x,t}\rho) 
  \end{align}
  where $\tilde{\rho}$ is again a Schwartz operator (and therefore trace class)
  according to Proposition \ref{prop:5}. 
\end{proof}

\subsection{Proof of the main theorem}
\label{sec:proof-main-theorem}

Based on the result derived so far we can now start to prove the main
theorem. This means in particular that we have to take the $x$-dependence
stronger into account. We start with a short lemma which will be the main tool
throughout this Subsection.

\begin{lem} \label{lem:3}
  For each $M \in \Bbb{N} \cup ~\{\infty\}$, $x \in [0,1]$, $t \in \Bbb{R}$
  with $|t|$ sufficiently small and each multiindex $S \in \{-1,1\}^d$
  (including the trivial case $|S| = d = 0$) let us define
  \begin{equation}
    \omega^S_{M,t,x}(\rho) = \tr \bigl(U_{M,t}^*(x)\,a^S\,U_{M,t}(x)\,\rho\bigr)
  \end{equation}
  and consider $\Sigma\subset\swa$ bounded, i.e.,
  \begin{equation}
    \|\rho\,(\no+p\Bbb{1})^{p/2}\|_1 < K_p\;\forall \rho\in\Sigma.
  \end{equation}
  Then we have
  \begin{equation}
    |\omega^S_{M,t,x}(\rho)| \leq K K_p
  \end{equation}
  with the constant $K$ from Lemma \ref{yab} and $p \in \Bbb{N}$ sufficiently
  big (such that the statement from Lemma \ref{yab} holds). Moreover
  $\omega^S_{M,t,x}(\rho)$ is continuous as a function of $x$.
\end{lem}

\begin{proof}
  Since $\rho$ is a Schwartz operator we can rewrite $\omega^S_{M,t,x}$ as
  \begin{equation}
    \omega^S_{M,t,x}(\rho) = \tr \bigl((\no+p\Bbb{1})^{-p/2} U_{M,t}(x)^* a^S
    U_{M,t}(x) \rho\,(\no+p\Bbb{1})^{p/2}\bigr)
  \end{equation}
  hence
  \begin{equation}
    |\omega^S_{M,t,x}(\rho)|\leq \|(\no+p\Bbb{1})^{-p/2} U_{M,t}(x)^* a^S
    U_{M,t}(x) \|\|\rho\,(\no+p\Bbb{1})^{p/2}\|_1 \leq K\cdot K_p\label{bdo}
  \end{equation}
  The first term is bounded by $K$ by Lemma \ref{yab}, the second term is
  bounded by $K_p$ by assumption (where $p$ has to be chosen large enough such
  that Lemma \ref{yab} is true). This shows boundedness.

  To show continuity in $x$ consider first the case $M = \infty$. For a
  quadratic Hamiltonian $\exp (-it H) a^S \exp (it H)$ is a linear
  differential operator in $x$ with coefficients smoothly depending on
  $t$. Hence the continuity of $\omega^S_{\infty,x,t}(\rho)$ in $x$ follows
  from that of $\exp(-i x t H)$ for each $t$ fixed.

  If $M$ is finite the situation is effectively finite dimensional, since
  $U_{M,t}(x)$ acts as the identity on all $\psi_n$ with $n>Mx$; the monomial
  $a^S$ can map $\psi_n$ to $\psi_{n+|S|}$; hence, we have to look at the
  Hilbert space spanned by $\psi_n,\, n=0,1,2,\dots,xM+|S|$. Now, continuity
  of $\omega^S_{M,x,t}(\rho)$ follows from the continuity of the map
  $x\mapsto H_M(x)$ in norm for all fixed $\lambda,M$, since this
  continuity implies norm continuity of $U_{M,t}(x)$ in $x\;\forall M,t$.
\end{proof}

\begin{lem}
  \label{ub}Consider a sequence of continuous maps: $R_{M,\lambda}:[0,1]\to
  \swa$, $M \in \Bbb{N}$ such that
  \begin{enumerate}
  \item \label{item:13}
    $\{R_{M,\lambda}(x)|x\in [0,1], M\in\mathbb{N}\}$ is a bounded subset
    of $\swa$
  \item $R_{M,\lambda}(x)$ converges in $\swa$ and uniformly on a
    neighbourhood $I\subset [0,1]$ of $\lambda$ to a continuous function
    $R_\infty : [0,1] \rightarrow \mathcal{S}(\mathcal{H}_\infty)$. 
  \end{enumerate}
  Then the sequence 
  \begin{equation}
    \tr \big(U_{M,t}^*(x)\,a^S\,U_{M,t}(x)\,R_{M,\lambda}(x) - U_{x,t}^*\,a^S\,U_{x,t}R_\infty(x)\big)
  \end{equation}
  is uniformly bounded on $x \in [0,1]$ and converges uniformly to $0$ on $x
  \in I$.
 \end{lem}

\begin{proof} 
  Convergence of the $R_{M,\lambda}(x)$
  implies that for each $p \in\Bbb{N}$ and each $\epsilon$ we can find an
  $M_{p,\epsilon}$ such that
  \begin{equation}
    \| (R_{M,\lambda}(x) - R_\infty(x)) (\no+p\Bbb{1})^{p/2} \| <
    \frac{\epsilon}{K}\quad \forall M > M_{p,\epsilon}
  \end{equation}
  holds. With Lemma \ref{lem:3} this implies
\begin{equation}
  |\omega^S_{M,t,x}(\RM-R_\infty(x))|\leq K\frac{\epsilon}{K}=\epsilon,\;\forall
  x\in I, \forall M>M_\epsilon
\end{equation}
  Hence, we write
\begin{multline}
  \bigl\lvert \tr \bigl(U_{M,t}^*(x)\,a^S\,U_{M,t}(x)\,R_{M,\lambda}(x) -
  U_{x,t}^*\,a^S\,U_{x,t}R_\infty(x)\bigr)\bigr\rvert \leq \\
  \bigl\lvert\omega^S_{M,t,x}\bigl(\RM-R_\infty(x)\bigr)\bigr\rvert +
  \bigl\lvert\omega^S_{M,t,x}\bigl(R_\infty(x)\bigr) -
  \omega^S_{\infty,t,x}(R_\infty(x))\bigr\rvert \ ,  
  \end{multline}
  where the first term was just shown to be $<\epsilon$ and the second term
  converges to zero by Proposition \ref{Ss}. Note, that both statement are
  independent of $x\in I$. 

  Boundedness follows similarly, if use condition \ref{item:13} and the fact
  that the map $R_\infty(x) \in \mathcal{S}(\mathcal{H}_\infty)$ is bounded as well
  (due to continuity and compactness of $[0,1]$).
\end{proof}

\begin{lem} \label{lem:4}
  Consider $M \in \Bbb{N} \cup \{\infty\}$, $x \in [0,1]$, $t \in \Bbb{R}$
  with $t$| sufficiently small, a multiindex $S \in \{-1,1\}^d$ and a
  continuous map $R: [0,1] \rightarrow \mathcal{S}(\mathcal{H}_\infty)$. Then
  the functions 
  \begin{equation}
    [0,1] \ni x \mapsto \omega^S_{M,x,t}\bigl(R(x)\bigr) \in \Bbb{C}
  \end{equation}
  are continuous.
\end{lem}

\begin{proof}
  Consider a sequence $x_j \in [0,1]$, $ j \in \Bbb{N}$ converging to $x
  \in [0,1]$. Continuity of $R$ implies that for each $\epsilon > 0$ and each
  $r \in \Bbb{N}$ there is an $j_{r,\epsilon}$ such that
  \begin{equation} \label{eq:32}
    \| (\no+r\Bbb{1})^{r} \rho \| < \epsilon 
  \end{equation}
  holds for all $\rho$ in the set
  \begin{equation}
    \Sigma_{r,\epsilon} = \{ R(x_j) - R(x) \, | \, j > j_{\epsilon,r} \}.
      \end{equation}
  Hence $j > j_{r,\epsilon}$ implies
  \begin{multline}
    \bigl\lvert|\omega^S_{M,x_j,t}\bigl(R(x_j)\bigr) -
    \omega^S_{M,x,t}\bigl(R(x)\bigr) \bigr\rvert \leq \\
    \bigl\lvert|\omega^S_{M,x_j,t}\bigl(R(x_j) - R(x)\bigr)\bigr\rvert +
    \bigl\lvert \omega^S_{M,x_j,t}\bigl(R(x)\bigl) -
    \omega^S_{M,x,t}\bigl(R(x)\bigr)\bigr\rvert .
  \end{multline}
  The first term on the right hand side is bounded by $K \epsilon$ due to
  (\ref{eq:32}) and Lemma \ref{lem:3}. The second term can be made arbitrarily
  small due to continuity of $\omega^S_{M,x,t}(\rho)$ in $x$; cf. again Lemma
  \ref{lem:3}. 
\end{proof}

\begin{lem} 
  \label{elozo} The sequence
  \begin{equation} \label{eq:29}
    M\mapsto(a_M(x)^S-x^\frac{|S|}{2}a^S)(\no+2|S|\Bbb{1})^{-|S|}\ .
  \end{equation}
    converges to zero in norm uniformly in $x \in [0,1]$
    any closed subinterval of $x\in I\subset [0,1]$ not containing $0$.
\end{lem}

\begin{proof}
  The statement is proved if we can show that for each
  $\epsilon > 0$ we can find an $M_\epsilon$ which is \emph{independent of $n$ and
  $x$} such that 
  \begin{equation} \label{eq:28}
    \| a_M^S(x)-x^\frac{|S|}{2}a^S)(\no+2|S|\Bbb{1})^{-|S|} \psi_n \| < \epsilon
    \epsilon 
  \end{equation} 
  holds for all $M > M_\epsilon$. Using the definition of $a_M(x)$ in
  Equations (\ref{b1}) and (\ref{b2}) and properties of the standard creation
  and annihilation operators it is easy to see that the elements of the
  sequence in 
  (\ref{eq:28}) are products of $|S|$ terms of the form
  \begin{equation}
    \left(\beta_{M,\lambda}(x,n+r)-\xpl\right)\sqrt{n+r'}\,(n+2|S|)^{-1}
  \end{equation}
  which differ only in the parameter $r, r'$; we have used a similar argument
  already in the proof of Lemma \ref{whichone}; cf. Equation
  (\ref{schem}). The only difference are $|S|$ terms $(n+2|S|)^{-1}$ we have
  distributed among the $S$ factors. Using the bound in lemma \ref{lem:2} we
  get
  \begin{equation}
    \frac{\beta_{M,\lambda}(x,n+r)-\xpl}{\sqrt{n+2|S|}}
    \frac{\sqrt{n+r'}}{\sqrt{n+2|S|}} \leq \frac{1}{\sqrt{M}}
    \sqrt{\frac{n+r}{n+2|S|} \frac{n+r'}{n+2|S|}}
  \end{equation}
  The right hand side can be bound in an $n$-independent way to 
  \begin{equation}
    \sqrt{\frac{n+r}{n+2|S|} \frac{n+r'}{n+2|S|}} \leq \sqrt{ \max(1,r/2S)
      \max(1,r'/2S) }.
  \end{equation}
  This leads to the estimate in (\ref{eq:28}) which concludes the proof.
 \end{proof}

 \begin{lem} \label{lem:6}
   For each $p \in \Bbb{N}$, $t \in \Bbb{R}$ sufficiently small and
   a uniformly bounded sequence $R_M: [0,1] \rightarrow
   \mathcal{S}(\mathcal{H}_\infty)$  the sequence  
   \begin{equation}
     \| (\no+2p\Bbb{1})^{p} U_{M,t}(x) R_M(x) U_{M,t}^*(x) \|_1 \quad M \in
     \Bbb{N} 
   \end{equation}
   is uniformly bounded in $x$.
 \end{lem}

 \begin{proof}
   Since $(\no+2p\Bbb{1})^{p}$ is a polynomial in $a$ and $a^*$ it is
   sufficient to look at
   \begin{equation}
     \| a^S U_{M,t}(x) R_M(x) U_{M,t}^*(x) \|_1 .  
   \end{equation}
   In addition we can rewrite this expression as
   \begin{align}
     \| a^S U_{M,t}(x) & R_M(x) U_{M,t}^*(x) \|_1  = \| U_{M,t}^*(x) a^S
     U_{M,t}(x) R_M(x)  \|_1 \\
     &= \| U_{M,t}^*(x) a^S  U_{M,t}(x) (\no+p'\Bbb{1})^{-p'/2}
     (\no+p'\Bbb{1})^{p'/2} R_M(x)  \|_1 \\
     &\leq \| U_{M,t}^*(x) a^S  U_{M,t}(x) (\no+p'\Bbb{1})^{-p'/2} \| \|
     (\no+p'\Bbb{1})^{p'/2} R_M(x) \|_1 \leq K' K_1
   \end{align}
   The first factor is bounded according to Lemma \ref{yab} the second because
   $R_M$ is uniformly bounded by assumption and $\rho \mapsto a^S \rho$ is a
   continuous map on $\mathcal{S}(\mathcal{H}_\infty)$; cf. Proposition
   \ref{prop:5}. 
 \end{proof}

 \begin{lem} \label{lem:7}
   For each $M \in \Bbb{N}$, $t \in \Bbb{R}$ sufficiently small and $R \in
   \{-1,1\}^d$ the function
   \begin{equation}
     [0,1] \ni x \mapsto \bigl\lvert \tr\bigl(U_{M,t}(x) R_M(x) U_{M,t}^*(x)
     (a_M^R(x) - a_\infty(x)^R)\bigr) \bigr\rvert \in \Bbb{C}
   \end{equation}
   is continuous.
 \end{lem}

 \begin{proof}
   We can write this expression as a linear combination of therms
   $\omega^S_{M,x,t}\bigl(R_M(x)\bigr)$ with coefficients depending continuously
   on $x$. Hence the statement follows from Lemma \ref{lem:4}.
 \end{proof}

Now we are ready to finish the proof of Theorem \ref{thm:1}. To this end let
us go back to Equation (\ref{eq:9}) which was
\begin{multline}
  \bigl\lvert \tr(U_{M,t} \rho_M U_{M,t}^* a_M^R) - 
    \tr(U_{\lambda,t} \rho_\infty U_{\lambda,t}^* a^R ) \bigr\rvert < \\
    \int_0^1 \bigl\lvert \tr\bigl(U_{M,t}(x) R_M(x) U_{M,t}^*(x) (a_M^R(x) -
    a_\infty(x)^R)\bigr) \bigr\rvert \mu_M(dx) + \\ \int_0^1 \bigl\lvert
    \tr\bigl((U_{\lambda,t} R_\infty(x) U_{\lambda,t}^* - U_{M,t}(x)
    R_\infty(x) U_{M,t}^*(x)) a_\infty(x)^R\bigr)\bigr\rvert \mu_M(dx) 
\end{multline}
where $a_\infty(x) = x^{|R|/2} a$ according to Equation (\ref{eq:31}). 

The integrand of the second integral is -- according to Lemmas \ref{ub} and
\ref{lem:4} -- continuous and uniformly bounded in $x \in [0,1]$, and it
converges uniformly on $x \in I$ to $0$. In addition we know that the measures
$\mu_M$ converge weakly to the point measure at $\lambda$. Hence the second
integral vanishes in the limit $M \rightarrow 0$. 

The integrand of the first integral is continuous as well (Lemma \ref{lem:7})
and we can rewrite it according to
\begin{multline}
  \bigl\lvert\tr\bigl((a_M^R(x) - x^{R/2} a^R) (\no+2|S|\Bbb{1})^{-|S|}
  (\no+2|S|\Bbb{1})^{|S|} U_{M,t}(x) R_M(x) U_{M,t}^*(x) \bigr)\bigr\rvert \\
  \leq \| (a_M^R(x) - x^{R/2} a^R) (\no+2|S|\Bbb{1})^{-|S|} \| 
  \| (\no+2|S|\Bbb{1})^{|S|} U_{M,t}(x) R_M(x) U_{M,t}^*(x) \|_1.
\end{multline}
The first factor on the right hand side converges to $0$ uniformly in $x$ (Lemma
\ref{elozo}) while the second is uniformly bounded (Lemma \ref{lem:6}). Hence in
the limit $M \rightarrow \infty$ this integral vanishes as well, which concludes
the proof.

\section{Outlook}
\label{sec:outlook}

We have shown that a one degree of freedom, continuous variable quantum system
can be simulated by mean field fluctuations of an ensemble of qubits. 
Our results substantially exceed existing schemes. This includes in particular
the range of states we can simulate and the treatment of the
dynamics. Nevertheless, there are a number of open questions which needs to be
treated in forthcoming papers.

\begin{itemize}
\item 
  \textbf{Dynamics.} Our results concerning the dynamics have two essential
  restrictions: The statements are restricted to quadratic Hamiltonians and to
  small times. The latter is probably only a restriction of the methods in the
  proof and can be removed with a more careful analysis of the limiting behavior
  of the unitaries $U_{M,t}(x)$. One possibility is to show that the time
  evolved density operators $U_{M,t}(x)R_M(x) U_{M,t}^*(x)$ converge in the same
  as the §$R_M(x)$ do, because in this case we can stack arbitrarily many
  finite time intervals together.
  
  The discussion of more general Hamiltonians is most likely a more difficult
  problem. In Section \ref{sec:herm-funct-analyt} we have made heavy use of
  the fact that Hermite functions are analytic vectors for \emph{all}
  quadratic Hamiltonians, and this argument has to be replaced somehow. A
  minimal requirement is in any case that the dynamics leaves the space of
  Schwartz functions invariant. The questions are: Which Hamiltonians have
  this property and is it sufficient to prove an analog of Theorem
  \ref{thm:1}? 
\item 
  \textbf{The reference state.} The reference state $\theta$ plays a very
  crucial role in the theory, which was overlooked in previous
  publications. It provides the parameter $\lambda$ which plays in some
  respect the role of an effective $\hbar$, although this is not universally
  true, as we have seen during the discussion of Theorem \ref{thm:2}. In
  particular the limit $\lambda \rightarrow 0$ seems to be boring, because in
  our scheme the observables $Q_\infty$ and $P_\infty$ just scale down with
  $\lambda^{-1/2}$ and therefore all expectation values vanishes if $\lambda$
  approaches $0$. However, this is not necessarily the end of the
  story. Instead of using sequences of states converging according to
  Definition \ref{def:2} we can use more divergent ones such that expectation
  values do not vanish in the limit $M \rightarrow \infty$. Alternatively we
  can rescale $Q_M(x)$ and $P_M(x)$ by $x^{-1/2}$ and $x^{1/2}$ respectively. 

  Another important aspect arises in connection dynamics. The time evolution
  we have looked at in this paper leave the references state invariant (at
  least asymptotically in the limit $M \rightarrow \infty$). It would be
  interesting to break this property and consider situations where the
  reference state is time-dependent.  
\item 
  \textbf{Classical observables.}
  Throughout our analysis we have basically ignored the fluctuation operator
  $F_M(\sigma_3)$. A short calculation along the lines of Section
  \ref{sec:fluct-oper} shows that $F_M(\sigma_3)$ commutes in the limit $M
  \rightarrow \infty$ with $Q_M$ and $P_M$, hence it describes a classical
  observables. The fact that $\tr(\sigma_3 \theta) = 0$ holds indicates that
  it is related to $\lambda$. 
\item 
  \textbf{Higher dimensions.}
  To extend of our schemes to systems with arbitrary, finite dimensional
  one-particle spaces is a very natural task, and the generalization of the
  discussion in Section \ref{sec:fluct-oper} is straightforward. The analysis
  of permutation invariant states along the lines of Section
  \ref{sec:perm-invar-stat} is not. This concerns in particular the embedding
  of irreducible representation spaces of unitary groups into Hilbert spaces
  $\mathrm{L}^2(\Bbb{R}^d)$.
\item 
  \textbf{Gaussian states.} 
  A lot of theory is available for Gaussian states and it is not completely
  clear how it is related to our approach. In \cite{Taku} it is shown that a
  translation invariant, exponentially clustering state $\omega$ of an
  infinite spin chain leads to Gaussian fluctuations if restrictions of
  $\omega$ to finite parts of the chain are used as the sequence $\rho_M$. In
  this case it clear already from the analysis of correlation functions (as in
  Section \ref{sec:fluct-oper}) that the operators $Q_\infty$ and $P_\infty$
  can be realized as canonical position and momentum in the Schrödinger
  representation (modulo factors of the form $\lambda^{\pm 1/2}$). Also the
  density operator $\rho_\infty$ is easily calculated since only one and two
  point correlation functions are needed. It is, however, not clear whether
  these sequences converge as described in Section \ref{sec:perm-invar-stat}. 

  Another interesting aspect concerns the relation between Gaussianity and
  correlations. The analysis in \cite{Taku} uses in a crucial way that the
  $\rho_M$ arises from restrictions of a thermodynamic limit state. Our
  scheme, on the other hand, does not need this assumption and it would be
  interesting up to which degree absence of long range correlations (in a to
  be specified limiting sense) still implies Gaussianity. Even more
  interesting (and more ambitious as well) is question whether we can
  distinguish classical correlations and entanglement by properties of the
  limit state $\rho_\infty$.\end{itemize}






\section*{Acknowledgments}

This work was supported by the European Commission under 
grants COQUIT (contract number 233747), ERC Grant GEDENTQOPT, and CHISTERA
QUASAR. We also like to thank the Spanish MINECO (Projects No.
FIS2009-12773-C02-02) the Basque
Government (Project No. IT4720-10), and the support
of the National Research Fund of Hungary OTKA (Contract No. K83858).

\begin{appendix}
  \section{Fluctuation operators and representations of tensor algebras}

  To prove Proposition \ref{prop:1} let us consider the tensor algebra
  $\mathcal{A}$ over the vector space $V = \Bbb{C}^2$:
  \begin{equation}
    \mathcal{A} = \bigoplus_{M \in \Bbb{N}} V^{\otimes M},\quad V^{\otimes 0} = \Bbb{C}.
  \end{equation}
  Together with the tensor product and the involution 
  \begin{equation}
    (x^{(1)} \otimes \dots \otimes x^{(k)})^* = \bar{x}^{(k)} \otimes \dots
    \otimes \bar{x}^{(1)},  
  \end{equation}
  (where $\bar{x}$ denotes complex conjugation in the canonical basis of $V =
  \Bbb{C}^2$) the space $\mathcal{A}$ becomes a *-algebra. The vector space $V$
  is naturally embedded in $\mathcal{A}$ and its canonical basis
  \begin{equation} \label{eq:18}
    q = \frac{(1,0)}{\sqrt{2}} \in V,\quad p = \frac{(0,1)}{\sqrt{2}} \in V
  \end{equation}
  is a complete set of generators. 

  The fluctuation operators $F_M$ give now rise to a *-representation of
  $\mathcal{A}$ by 
  \begin{equation}
    \Phi_M(q) = Q_M,\quad \Phi_M(p) = P_M
  \end{equation}
  where $Q_M = 2^{-1/2}F_M(\sigma_1)$, $P_M = 2^{-1/2} F_M(\sigma_2)$ are the
  operators defined in Equation (\ref{eq:7}). Any density operator $\rho_M \in
  \mathcal{B}(\mathcal{H}^{\otimes M})$ leads to a state 
  \begin{equation}
    W_M(X) = \tr(\Phi_M(X) \rho_M),\quad X \in \mathcal{A}
  \end{equation}
  of the algebra and if the sequence $\rho_M$, $M \in \Bbb{N}$ has $\sqrt{M}$
  fluctuations the limit
  \begin{equation}
    W(X) \lim_{M\rightarrow\infty} W_M(X) = \tr(\Phi_M(X) \rho_M)
  \end{equation}
  exists for all $X$ and defines again a state $W$ of $\mathcal{A}$. 
  
  The state $W$ defines the GNS representation\footnote{Usually the GNS
    representation is only considered for C*-algebras, but it can easily be
    generalized to more general *-algebras. We only have to allow
    representations in terms of unbounded operators. Our case is basically an
    instance of the celebrated reconstruction theorem of Wightman quantum field
    theory \cite{StreatWight}.} $(\mathcal{H}_W,D_W,\pi_w,\Omega_W)$ of
  $\mathcal{A}$ consisting of a Hilbert space $H_W$, a dense subspace $D_W$, a
  *-morphism $\pi_W$ into the algebra of (unbounded) operators $D_W \rightarrow
  D_W$ and a cyclic vector $\Omega_W$ satisfying
  \begin{equation}
    \langle \Omega_W ,\pi_w(X) \Omega_W\rangle = W(X).
  \end{equation}
  Without loss of generality we will assume in the following that $D_W$
  coincides with the subspace generated by vectors $\pi_W(X) \Omega_W \in
  D_W$. Applying $\pi_W$ to the generators $q,p$ from  Equation (\ref{eq:18})
  we get two (unbounded) operators 
  \begin{equation}
    Q_\infty = \pi_W(q),\quad P_\infty = \pi_W(p)
  \end{equation}
  with domain $D_W$. For any polynomial $f(q,p)$ they satisfy
  \begin{equation}
    \lim_{M \rightarrow \infty} \tr\bigl(f(Q_M,P_M) \rho_M\bigr) = \tr
    \bigl(f(Q_\infty, P_\infty) \rho_\infty\bigr)\quad \text{with}\quad
    \rho_\infty = \kb{\Omega_W}.
  \end{equation}
  This establishes all statements in Proposition \ref{prop:1} except the last
  commutation relations (item \ref{item:4}). 

  To show the latter, let us consider two elements $X_1,X_2 \in \mathcal{A}$ and
  the expression (cf. Equation (\ref{eq:10}))
  \begin{multline}
    \tr\bigl(\Phi_M(X_1) [Q_M,P_M] \Phi_M(X_2)\bigr) -i \lambda
    \tr\bigl(\Phi_M(X_1) \Phi_M(X_2) \rho_M\bigr) \\ = \frac{i}{2\sqrt{M}} 
    \tr\bigl(\Phi_M(X_1) F_M(\sigma_3) \Phi_M(X_2)\bigr).
  \end{multline}
  If the sequence $\rho_M$ has $\sqrt{M}$ fluctuations we have
  \begin{equation}
    \lim_{M \rightarrow \infty} \tr\bigl(\Phi_M(X_1) F_M(\sigma_3)
    \Phi_M(X_2)\bigr) = C < \infty,
  \end{equation}
  hence
  \begin{multline}
    W\bigl(X_1 (qp-pq-i\lambda\Bbb{1}) X_2\bigr) \\ = \lim_{M \rightarrow \infty}
    \bigl[ \tr\bigl(\Phi_M(X_1) [Q_M,P_M] \Phi_M(X_2)\bigr) -i \lambda
    \tr\bigl(\Phi_M(X_1) \Phi_M(X_2) \rho_M\bigr) \bigr] = 0.
  \end{multline}
  In other words, the limiting state $W$ vanishes on the ideal generated by
  $qp-pq-i\lambda\Bbb{1}$ which leads to 
  \begin{equation}
    [Q_\infty,P_\infty] \phi = i\lambda\phi \quad \forall \phi \in D_W
  \end{equation}
  as stated.
\end{appendix}

\end{document}